\documentclass{article}

\usepackage{arxiv}

\usepackage[utf8]{inputenc} 
\usepackage[T1]{fontenc}    
\usepackage{hyperref}       
\usepackage{url}            
\usepackage{booktabs}       
\usepackage{amsfonts}       
\usepackage{nicefrac}       
\usepackage{microtype}      
\usepackage{xcolor}         

\usepackage{amsmath,amsfonts}
\usepackage{amsthm,bm}
\usepackage{booktabs}
\usepackage{xcolor}
\usepackage{bm}

\usepackage{thmtools,thm-restate}

\usepackage{algorithm}
\usepackage{algpseudocode}

\usepackage{graphicx}
\usepackage{subcaption}

\usepackage{float}
\usepackage{multirow}
\usepackage{placeins}
\usepackage{lipsum}

\newtheorem{theorem}{Theorem}
\newtheorem{definition}{Definition}

\newtheorem{proposition}{Proposition}
\newtheorem{lemma}{Lemma}

\newtheorem{corollary}{Corollary}

\DeclareMathOperator*{\argmax}{arg\,max}
\DeclareMathOperator*{\argmin}{arg\,min}

\title{A Game-Theoretic Approach to Privacy-Utility Tradeoff in Sharing Genomic Summary Statistics}

\author{%
  Tao Zhang\thanks{Corresponding author} \\
    Computer Science \& Engineering\\
  Washington University in St. Louis\\
  St. Louis, MO 63130, USA \\
  \texttt{tz636@nyu.edu} \\
  \And
  Rajagopal Venkatesaramani\\
  Khoury College of Computer Sciences\\
  Northeastern University\\
  Boston, MA 02115, USA \\
  \texttt{r.venkatesaramani@northeastern.edu} \\
  \AND
  Rajat K. De \\
  Machine Intelligence Unit\\
  Indian Statistical Institute\\
Kolkata 700108, India. \\
  \texttt{rajat@isical.ac.in} \\
  \And
  Bradley A. Malin\\
  Biomedical Informatics\\
  Vanderbilt University\\
  Nashville, TN 37232, USA\\
  \texttt{b.malin@vumc.org} \\
  \And
  Yevgeniy Vorobeychik \\
  Computer Science \& Engineering\\
  Washington University in St. Louis\\
  St. Louis, MO 63130, USA \\
  \texttt{yvorobeychik@wustl.edu} \\
}

\begin{document}

\maketitle

\begin{abstract}
The advent of online genomic data-sharing services has sought to enhance the accessibility of large genomic datasets by allowing queries about genetic variants, such as summary statistics, aiding care providers in distinguishing between spurious genomic variations and those with clinical significance. 
However, numerous studies have demonstrated that even sharing summary genomic information exposes individual members of such datasets to a significant privacy risk due to membership inference attacks.
While several approaches have emerged that reduce privacy risks by adding noise or reducing the amount of information shared, these typically assume non-adaptive attacks that use likelihood ratio test (LRT) statistics.
We propose a Bayesian game-theoretic framework for optimal privacy-utility tradeoff in the sharing of genomic summary statistics.
Our first contribution is to prove that a very general Bayesian attacker model that anchors our game-theoretic approach is more powerful than the conventional LRT-based threat models in that it induces worse privacy loss for the defender who is modeled as a von Neumann-Morgenstern (vNM) decision-maker.
We show this to be true even when the attacker uses a non-informative subjective prior.
Next, we present an analytically tractable approach to compare the Bayesian attacks with arbitrary subjective priors and the Neyman-Pearson optimal LRT attacks under the Gaussian mechanism common in differential privacy frameworks.
Finally, we propose an approach for approximating Bayes-Nash equilibria of the game using deep neural network generators to implicitly represent player mixed strategies.
Our experiments demonstrate that the proposed game-theoretic framework yields \emph{both} stronger attacks and stronger defense strategies than the state of the art.
\end{abstract}


\section{Introduction}

In recent years, genomic sequencing for medical purposes as well as in direct-to-consumer services like ancestry discovery has become common.
This increase, however, has brought about substantial privacy concerns regarding the sharing of genomic data.
Due to these concerns, genomic information is typically shared publicly in summary form.
For example, the Global Alliance for Genomics and Health (GA4GH) introduced the genomic data-sharing standard named \textit{genomic data-sharing beacon protocol}, in which queries involving \emph{single-nucleotide variants (SNVs)} involve presence or absence of an alternate or minor allele in a dataset \cite{global2016federated}.
A common alternative that provides somewhat more information is to share \emph{summary statistics} for SNVs---that is, fraction of individuals in the dataset who have the alternate allele (also known as \emph{alternate allele frequencies (AAFs)}\cite{sankararaman2009genomic,macarthur2021workshop}.

Nevertheless, a series of studies have demonstrated that summary genomic data sharing is insufficient to protect the privacy of individuals in the dataset\cite{homer2008resolving,sankararaman2009genomic,shringarpure2015privacy,ayoz2021genome,bu2021haplotype,samani2015quantifying,von2019re}.
In particular, these studies have shown that it is possible for a malicious actor in possession of a target genome to make accurate inferences about membership of the associated individual in such a database using only the released summary information by leveraging likelihood ratio test statistics (LRT).
In response to such \emph{membership inference attacks (MIAs)}, several approaches have been proposed to improve the privacy of released genomic summary statistics, including differentially private mechanisms \cite{wang2014differentially,yilmaz2022genomic}, falsifying responses for rare alleles \cite{raisaro2017addressing,wan2017controlling}, optimization techniques \cite{venkatesaramani2023defending,venkatesaramani2023enabling}, and game-theoretic approaches \cite{wan2017expanding}.

Most existing approaches for defending against membership inference attacks (MIAs) assume a non-adaptive likelihood ratio test (LRT) attacker. An alternative to these methods is differential privacy (DP), which has the advantage of making few assumptions about the nature of the attack. However, the privacy parameter for DP is often selected for specific use cases without prior evaluation of its impact \cite{loukides2010anonymization,tramer2015differential}.
This approach can lead to coarse trade-offs between privacy and utility, potentially resulting in either insufficient protection or excessive sacrifice of data utility.
We address these limitations by modeling the privacy problem as a Bayesian game with a bounded rational attacker.
In our game model, the attacker's bounded rationality manifests as a \emph{subjective prior} about dataset membership, which may be inconsistent with the objective membership distribution.
We assume that the defender---such as a major medical center or a consortium as in the case of GA4GH---has better information about the membership prior, and in this regard has an informational advantage.
Moreover, we treat the actual data as the defender's \emph{private information}.
Thus, the defender's strategy maps a dataset to a probability distribution over summary statistics that are released, whereas the attacker's strategy maps summary data release to a decision about membership claims.
We assume that the defender incurs a cost that stems from any noise added in the released summary statistics, whereas the attacker's cost is linear in the number of positive membership claims.

Our main theoretical result compares our bounded-rational Bayesian attacker to the common LRT-based threat model in terms of the worst-case loss incurred by the defender. Specifically, we show that under mild assumptions, even a bounded-rational Bayesian attacker with a non-informative prior results in a higher worst-case loss for the defender than the best LRT attack, which is the uniformly most powerful (UMP) test for any significance level according to the Neyman-Pearson lemma \cite{neyman1933ix}.
Additionally, we provide a tractable method to quantitatively compare the Bayesian attack and the LRT attack when the Bayesian attacker has arbitrary subjective priors within a class of Gaussian defense mechanisms. We derive the necessary and sufficient conditions under which Bayesian attacks with arbitrary subjective priors outperform or underperform UMP LRT attacks in terms of the number of SNVs involved in the Beacon dataset.

Our second contribution is a novel approach for approximating a Bayes-Nash equilibrium of the game by representing the defender mixed strategies as a neural network generator, whereas the attacker strategies are represented as a neural network binary classifier (since the attacker can decide whether to attack each potential dataset member independently).
We then use gradient-based methods to effectively solve the game for any parameters trading off the privacy risks and utility loss due to added noise for the defender.
Our experiments on a genomic dataset provided by the 2016 iDASH privacy challenge demonstrate that the proposed approach significantly outperforms prior art in terms of \emph{both} efficacy of defense and efficacy of attacks.

\noindent\textbf{Related Work: }
In a seminal study in 2008, Homer et al. \cite{homer2008resolving} uncovered the possibility of inferring identities of individuals in genomic databases solely through the allele frequencies of the participants  (i.e., the summary statistics).
Subsequently, an extended spectrum of vulnerabilities for membership inference attacks has been explored \cite{erlich2014}, underscoring the significant privacy risks associated with the public release of summary statistics from genetic databases.
Our work considers threat models within the context of membership attacks aimed at summary data releases \cite{shringarpure2015privacy,raisaro2017addressing,von2019re,ayoz2020effect,ayoz2021genome}.
Typically, the threat models considered in such scenarios are framed as LRT attacks \cite{sankararaman2009genomic,shringarpure2015privacy,raisaro2017addressing,venkatesaramani2021defending,venkatesaramani2023enabling}.
Our approach to privacy protection aligns recent frameworks for optimizing privacy-utility tradeoff in the context of LRT attacks \cite{venkatesaramani2023enabling,venkatesaramani2021re}, as well as differential privacy \cite{cho2020privacy,dwork2006calibrating,dwork2006differential} which perturb responses (e.g., adding noise, flipping) before their release.
Venkatesaramani et al. \cite{venkatesaramani2023enabling} offer state-of-the-art approach for defending against membership inference on summary statistics releases.
However, they consider a bounded-rational attacker leveraging solely LRT for membership inference with either fixed and adaptive thresholds.
However, we find that the resulting defense mechanisms against LRT attacks are not robust against even bounded-rational Bayesian attackers who can induce higher privacy loss for the defender under the same defense mechanism.
Finally, our work is closely related to game-theoretic approaches to privacy-preserving data-sharing \cite{do2017game,freudiger2009non,wan2015game}, including in the context of sharing genomic summary data \cite{wan2017expanding,wan2020game,wan2021using}.
Most such approaches, however, treated the problem as a complete information game, and address either a very different domain (e.g., location privacy) or very different privacy mechanisms (e.g., suppression rather than noise addition).

\section{Preliminaries}

Consider a universe (population) $U\equiv[K]$ of $K$ individuals containing their genomic information in the form of \emph{single-nucleotide variants (SNVs)}.
A particular SNV $j$ for an individual $k$ is commonly coded using a binary value which indicates whether or not this individual has an alternate allele for SNV $j$.
Thus, we represent the data of an individual $k$ in the population as a binary vector $d_k$, with entries corresponding to SNVs with $d_{kj}=0$ whenever $k$ has a minor allele $j$, and $d_{ij}=1$ otherwise.
Let $Q$ be the set of all SNVs under consideration.

The goal is to share \emph{summary statistics} for a private database $B$ of individuals.
To formalize, we assume that $B \subseteq U$.
Consequently, we can use a vector $b \in W \equiv \{0,1\}^K$ to denote a binary membership vector associated with $B$, where $b_k=1$ iff individual $k$ is in $B$.
Further, we assume that membership in $B$ (that is, the vector $b$) is generated stochastically according to a prior probability distribution $\theta \in\Delta(W)$ which is common knowledge (this assumption is important in formalizing the game-theoretic model).
Finally, summary statistics of $B$ constitute a vector of alternate allele frequencies (AAFs) for for set $Q$ of SNVs, denoted by $x=(x_{1},\dots, x_{|Q|})=f(b,d)\in \mathcal{X}\equiv[0,1]^{|Q|}$, where $f$ computes the fraction of individuals who do not have a minor allele for each SNV, that is, $x_j = f_j(b,d) \equiv \frac{1}{\sum_k b_k}\sum_k b_k d_{kj}$ for all $j$.
Throughout the paper, we assume SNVs are independent, ensured by a prefiltering protocol that retains a subset in linkage equilibrium \cite{kimura1965attainment}.



Our primary concern is with \emph{membership inference attacks (MIAs)} that use summary statistics to determine whether particular individuals $k \in U$ are in the private data $B$.
Let $s=(s_{k})_{k\in U}$ be the vector of binary predictions about membership, where $s_k \in \{0,1\}$ denotes the prediction whether $k$ is in $B$ or not.
Given $b$ and $s$, we capture privacy loss of each individual $k$ by $v_{k}(s_{k},b_{k})$, which satisfies $v_{k}(1,1)>v_{k}(0,0)\geq v_{k}(1,0)=v_{k}(0,1)\geq 0$.
In this work, we consider a simple representative formulation $v_{k}(s_{k}, b_{k})\equiv s_{k}b_{k}$.
Hence, the total privacy loss is given by $v(s,b)\equiv\sum\nolimits_{k\in U} v_{k}(s_{k}, b_{k})=\sum_{k\in U} s_{k}b_{k}$. 
We refer to the data curator of the private dataset $B$ as the \emph{defender} who aims to protect the privacy of $k \in B$ against membership inference attacks. 
We refer to the agent who undertakes such inference as the \emph{attacker}.
The total privacy loss $v(s,b)$ is what the defender aims to minimize, while the attacker wishes to maximize.

\paragraph{Data Protection Model}

Our approach to privacy protection entails releasing perturbed summary statistics, which follows the common paradigms of randomized response (such as common differential privacy frameworks) in prior literature \cite{blair2015design,dwork2006calibrating,dwork2006differential,venkatesaramani2023enabling}.
Specifically, the defender adds a noise $\delta=(\delta_{j})_{j\in Q}\in D$ to AAFs (summary statistics) $x$, so that $r=\mathtt{R}(x+\delta)\in [0,1]$ is released as the (perturbed) summary statistics, where typically $\mathtt{R}(x+\delta)\equiv \mathrm{Clip}_{[0,1]}(x+\delta)$ to ensure that $r\in[0,1]^{|Q|}$.

\paragraph{Likelihood Ratio Test Attacks}

MIAs targeting genomic summary data releases commonly treat it as a hypothesis test problem \cite{sankararaman2009genomic,shringarpure2015privacy,raisaro2017addressing,venkatesaramani2021defending,venkatesaramani2023enabling}: $H^{k}_{0}: b_{k}=1$ versus $H^{k}_{1}: b_{k}=0$, for each $k\in U$.
In addition, $\bar{p}_{j}$ is the frequency of the alternate allele at $j$-th SNV in a reference population of individuals who are not in $B$.
First, suppose that $\delta = 0$.
The attacker is assumed to possess external knowledge about the genome of individuals in $[K]$ in terms of $\bar{p}=(\bar{p}_{j})_{j\in Q}$ and $d=(d_{kj})_{k\in [K], j\in Q}$.
The log-likelihood ratio statistics (LRS) for each individual $k$ is given by \cite{sankararaman2009genomic}:
\begin{equation}\label{eq:standard_LRT}
    \begin{aligned}
        \ell(d_{k}, x) = \sum\nolimits_{j\in Q}\left( d_{kj} \log\frac{\bar{p}_{j}}{x_{j}} + (1-d_{kj}) \log\frac{1-\bar{p}_{j}}{1-x_{j}}\right).
    \end{aligned}
\end{equation}
An LRT attacker performs MIA by testing $H^{k}_{0}$ against $H^{k}_{1}$ using $\ell(d_{k}, x)$ for every $k\in[K]$. 
The null hypothesis $H^{k}_{0}$ is rejected in favor of $H^{k}_{1}$ if $\ell(d_{k}, x)\leq\tau$ for a rejection rule $\tau^{k}$, and $H^{k}_{0}$ is accepted if $\ell(d_{k}, x)>\tau$.
Let $P^{k}_{0}(\cdot)\equiv\textup{Pr}(\cdot|H^{k}_{0})$ and $P_{1}(\cdot)\equiv\textup{Pr}(\cdot|H^{k}_{1})$ denote the probability distributions associated with $H_0$ and $H_1$, respectively.
\begin{definition}[$(\alpha_{\tau},\beta_{\tau})$-LRT Attack]
    The attacker performs \textup{$(\alpha_{\tau},\beta_{\tau})$-LRT Attack} if $P^{k}_{0}(\ell(d_{k}, x)\leq\tau)= \alpha_{\tau}$ and $1-P^{k}_{1}(\ell(d_{k}, x)\leq\tau)=\beta_{\tau}$, for all $k\in U$, where $\alpha_{\tau}$ is the \textup{significance level} and $1-\beta_{\tau}$ is the \textup{power} of the test with the \textup{threshold} $\tau$.
\end{definition}
Define the trade-off function \cite{dong2021gaussian}, $T[P^{k}_{0}, P^{k}_{1}](\alpha) \equiv \inf_{\tau}\{\beta_{\tau}: \alpha_{\tau}\leq \alpha\}$.
By Neyman-Pearson lemma \cite{neyman1933ix}, the LRT test is the uniformly most powerful (UMP) test for a given significance level.
In particular, for a given significance level $\alpha_{\tau^{*}}$, there exists a LRT with threshold $\tau^{*}$ such that no other hypothesis test with $\alpha\leq \alpha_{\tau^{*}}$ can achieve a strictly smaller $\beta_{\tau}< \beta_{\tau^{*}}$. 
Hence, $T[P^{k}_{0}, P^{k}_{1}](\alpha_{\tau^{*}})=\beta_{\tau^{*}}$, for all $k\in U$.
We use $\alpha$-LRT to denote a UMP $(\alpha,\beta)$-LRT and interchangeably add and drop the notation of the corresponding threshold.



\section{Genomic Privacy Protection as a Bayesian Game}

\subsection{Von Neumann-Morgenstern Defender}

The trade-off between privacy and utility (or data quality) inevitably constrains the defender's choice of noise to perturb the summary statistics.
We model the utility degradation due to the perturbation using the noise $\delta=(\delta_{j})_{j\in Q}\in D$.
Thus, when the true membership is $b$, the added noise is $\delta$, and $s$ is inferred, the cost of the defender is given by
\begin{equation*}
    u_{D}(\delta, b, s)\equiv v(s,b) + \sum\nolimits_{j\in Q}\kappa_{j}|\delta_{j}|,
\end{equation*}
where $\kappa_{j}\geq 0$ represents the defender's preference over the privacy-utility trade-off given the summary statistics of the $j$th SNVs, for all $j\in Q$.

In this work, we consider when the defender is a Von Neumann-Morgenstern (vNM) decision-maker, who deals with privacy risks by acting to minimize the expected privacy loss.
In particular, let $g_{D}:W\mapsto \Delta(D)$ denote the defense mechanism, so that $g_{D}(\delta|b)$ specifies the probability of adding a noise $\delta$ to obtain $r=\mathtt{R}(x+\delta)$.
The probability distribution $\rho_{D}(\cdot|b)\in \Delta([0,1]^{|Q|})$ of the random variable $\tilde{r}$ on $[0,1]^{|Q|}$ is uniquely determined by $g_{D}$; i.e., 
$\tilde{r}=\mathtt{R}(f(b,d)+\tilde{\delta})\sim\rho_{D}(\cdot|b)$ if and only if $\tilde{\delta}\sim g_{D}(\cdot|b)$.
The randomness of $g_{D}$ leads to an expected utility cost denoted by $\mathbb{E}[\|\delta\||g_{D}]$.
Suppose $g_{D}$ is a Laplace distribution centered at zero, i.e., $\textup{Laplace}(0, \mathtt{scale})$. 
Then, the defense mechanism $g_{D}$ is $\epsilon$-differentially private \cite{dwork2006differential}, where $\epsilon=\frac{\mathtt{sensitivity}}{\mathtt{scale}}$ with $\mathtt{sensitivity}$ as the \textit{sensitivity} of $f$.
It is easy to verify that $\mathbb{E}[\|\delta\||g_{D}]=\mathtt{scale}=\frac{\mathtt{sensitivity}}{\epsilon}$.
Let $\mathtt{sens}$. 
Thus, in the Laplace defense mechanism, our formulation of the utility loss reflects the privacy-utility trade-off of the standard differential privacy: the decreasing (resp. increasing) $\epsilon$ leads to the increasing (resp. decreasing) of utility loss.

From a vNM defender's perspective, the expected privacy losses under an $\alpha$-LRT attack, without and with defense $g_{D}$, respectively, are given by
\begin{equation*}
    \begin{aligned}
        &L^{o}(\tau^{o}, \alpha)\equiv\mathbb{E}\left[v(\tilde{s}, \tilde{b})\middle|\alpha\right]= \sum\nolimits_{k}P^{k}_{1}\left[y_{k}(f(b,z),\tau^{o})=1\right] \theta(b_{k}=1)=\sum\nolimits_{k}(1-\beta^{0}_{\tau})\theta(b_{k}=1),\\
&L(g_{D},\tau^{o},\alpha)\equiv\mathbb{E}\left[v(\tilde{s}, \tilde{b})\middle|g_{D},\tau^{o},\alpha\right]= \sum\nolimits_{k}P^{k}_{1}\left[y_{k}(r,\tau^{o})=1\middle|g_{D}\right] \theta(b_{k}=1),
    \end{aligned}
\end{equation*}
where $y_{k}(x,\tau^{o})\equiv \mathbf{1}\left\{\ell(d_{k}, x)\geq \tau^{o}\right\}$, $\mathbf{1}\{\cdot\}$ is the indicator function,  $P^{k}_{1}[y_{k}(r,\tau^{o})=1|g_{D}]\equiv \int_{r} \mathbf{1}\left\{y_{k}(r,\tau^{o})=1\right\} \rho_{D}(r|b)dr$, and $\tau^{o}$ is the threshold associated with the $\alpha$-LRT.

\paragraph{Fixed-Threshold LRT Attack \cite{sankararaman2009genomic,shringarpure2015privacy,venkatesaramani2021defending,venkatesaramani2023enabling}}
An attacker can be seen as \textit{naive or credulous} if the attacker performs MIA without considering the existence of the defense protecting the privacy.
The naive attacker chooses a threshold $\tau^{o}$ that balances the Type-I and Type-II errors that lead to a UMP $\alpha$-LRT test when there is no defense.
The approximation of a UMP $\alpha$-LRT test may be obtained by simulating Beacons on datasets that are publicly available or are synthesized based on the knowledge of AAFs \cite{venkatesaramani2023enabling}.
The defender's optimal strategy against the naive $\alpha_{\tau^{o}}$-LRT attack solves the following problem 
\begin{equation}\tag{\texttt{NaiveLRT}}\label{eq:nariveLRT}
    \min\nolimits_{g_{D}} L(g_{D}, \tau^{o},\alpha_{\tau^{o}}) + \kappa U(g_{D}, \tau^{o},\alpha_{\tau^{o}}),
\end{equation}
where $U(g_{D}, \tau^{o},\alpha_{\tau^{o}})\equiv \mathbb{E}\left[\|\tilde{\delta}\|\middle|g_{D}, \tau^{o}, \alpha_{\tau^{o}}\right]$ is induced expected utility loss.

Let $\beta^{k}(\tau, g_{D}, \alpha)\equiv 1-P^{k}_{1}[y_{k}(r,\tau)=1|g_{D}]$ denote the actual Type-II error of the naive (UMP) $\alpha$-LRT attack with $\tau$ under the defense mechanism $g_{D}$.
Thus, the defender can decrease the (expected) privacy loss under the naive $\alpha$-LRT attack by choosing $g_{D}$ to increase $\beta^{k}(\tau, g_{D}, \alpha)$ for all $k\in U$.
Obviously, the efficacy of $g^{\dagger}_{D}$ that solves (\ref{eq:nariveLRT}) to attain a privacy loss $L(g^{\dagger}_{D}, \tau, \alpha)$  requires the defense implementation to be \textit{stealthy}.

\paragraph{Adaptive-Threshold LRT Attack \cite{venkatesaramani2021defending,venkatesaramani2023enabling}}
The attacker using \textit{adaptive-threshold LRT} attacks is aware of the implementation of the defense. 
The attacker attempts to separate $U$ from those who are in the reference population $\bar{U}$ (individuals not in $U$).
Let $\bar{U}^{(N)}\subset \bar{U}$ denote the set of $N$ individuals in $\bar{U}$, who have the lowest LRSs. 
The \textit{adaptive threshold} is given by $\tau^{(N)}(r)=\frac{1}{N}\sum_{i\in \bar{U}^{{N}}} \ell(d_{i}, r)$. 
Hence, the hypothesis $H_{0}$ is rejected if $\ell(d_{k}, r)\leq \tau^{(N)}(r)$.
Thus, the defender's problem 
\begin{equation}\tag{\texttt{AdaptLRT}}\label{eq:adaptiveLRT}
    \min\nolimits_{g_{D}} L(g_{D}, \tau^{(N)}(r),\alpha_{\tau^{(N)}(r)}) + \kappa U(g_{D}, \tau^{(N)}(r),\alpha_{\tau^{(N)}(r)}),
\end{equation}
where $\alpha_{\tau^{(N)}(r)}$ is the Type-I error attained by the using $\tau^{(N)}(r)$.

\paragraph{Optimal LRT Attack}
Denote by $P^{k}_{0}(g_{D})=P^{k}_{0}[\cdot|g_{D}]$ and $P^{k}_{1}(g_{D})=P^{k}_{1}[\cdot|g_{D}]$ the probability distributions associated with $g_{D}$ under $H^{k}_{0}$ and $H^{k}_{1}$, respectively. 
Then, the \textit{worse-case privacy loss} (WCPL) of the defender is attained when the attacker's hypothesis test achieves $\beta^{k}(\tau^{*}, g_{D}, \alpha)=T[P^{k}_{0}(g_{D}), P^{k}_{1}(g_{D})](\alpha)$ for some threshold $\tau^{*}$, which is a UMP test under the defense $g_{D}$.
We refer to such attack models as \textit{optimal $\alpha$-LRT attacks}.
Thus, the defender's optimal strategy against the optimal $\alpha$-LRT attack solves the following optimization problem:
\begin{equation}\tag{\texttt{OptLRT}}\label{eq:optimalLRT}
    \begin{aligned}
        \min\nolimits_{g_{D}} L(g_{D}, \tau^{*}, \alpha)+ \kappa U(g_{D}, \tau^{*},\alpha), \textup{ s.t. } \beta^{k}(\tau^{*}, g_{D}, \alpha)=T[P^{k}_{0}(g_{D}), P^{k}_{1}(g_{D})](\alpha).
    \end{aligned}
\end{equation}
By the Neyman-Pearson lemma, $\alpha$-LRT with LR statistics $\ell(d_{k},r;g_{D})\equiv\sum_{j\in Q}\frac{\rho_{D}(r|b_{k}=0,b_{-k})}{\rho_{D}(r|b_{k}=1,b_{-k})}$ for all $k\in U$ is optimal one that attains $\beta^{k}(\tau^{*}, g_{D}, \alpha)=T[P^{k}_{0}(g_{D}), P^{k}_{1}(g_{D})](\alpha)$.
In addition, it is not hard to verify that the $g_{D}$ obtained by solving (\ref{eq:optimalLRT}) is robust against the adaptive-threshold LRT attack.

\subsection{\texorpdfstring{$\sigma$}{Lg}-Bayesian Attacks}

In this section, we consider a strategic Bayesian attacker, who is also a vNM decision-maker.
The external knowledge of the attacker is captured by the subjective prior beliefs about $b\in W$, denoted by $\sigma\in \Delta(W)$.
We refer to such attacks as $\sigma$-Bayesian attacks.
The attacker launches MIA, aiming to infer $b\in W$ by obtaining an inference outcome $s\in W$.
The defender's privacy loss $v(s,b)$ is a privacy value for the attacker. 
The attacker may also encounter trade-offs between extracting as more amount of privacy and the operational costs induced by any post-processing to extract the value of private membership information, such as promoting personalized medication and marketing.
Hence, the final conclusions $s_{k}=1$ indicates two key outcomes: \textit{(i)} the individual $k$ is identified as a member of the set $B$, and \textit{(ii)} the attacker proceeds with the post-processing action on individual $k$.
The post-processing operation cost is then given by $c_{A}(s)=\sum_{k\in U} s_{k}$.
Therefore, with such a trade-off between the privacy value and the operation costs, the attack can be seen as a constrained membership inference.
When the true membership is $b$ and the attacker's inference is $s$, the cost function of the attacker is given by
\begin{equation*}
    u_A(s,b) = -v(s,b) + \gamma c_A(s),
\end{equation*}
where $\gamma\geq 0$ captures the attacker's preference over the trade-off of the privacy value and the operational cost.
The attacker adopts a mixed strategy, denoted by $h_{A}:\Gamma\mapsto\Delta(W)$, which specifies a probability distribution over $W$ based on an observation $r$ (i.e., perturbed summary statistics).

When the defender adopts $g_{D}$ (which induces $\rho_{D}$) and the attacker adopts $h_{A}$, the \textit{ex-ante} expected costs of the defender and the attacker are given by, respectively,
\begin{equation}
    \begin{aligned}
        &U_{D}\left(g_{D}, h_{A};\theta\right)\equiv \sum\nolimits_{s,b}\int\nolimits_{r} u_{D}(r, s,b)h_{A}(s|r)\rho_{D}(r|b)dr\theta(b),\\
        &U_{A}\left(g_{D}, h_{A}; \sigma\right)\equiv \sum\nolimits_{s,b}\int_{r}u_{A}(s,b)h_{A}(s|r)\rho_{D}(r|b)dr \sigma(b),
    \end{aligned}
\end{equation}
where $\theta\in\Delta(W)$ and $\sigma\in \Delta(W)$ are the true prior and the attackers' subjective prior, respectively.
Furthermore, any observation $r\in\Gamma$ can be used as evidential data for the attacker to update his subjective prior $\sigma$ to form a posterior belief according to Bayes' law: $\mu_{\sigma}(b|r) = \rho_{D}(r|b) \sigma(b)/\sum\nolimits_{b'\in W} \rho_{D}(r|b')\sigma(b')$,
which leads to more informed decision-making of the attacker to choose $h_{A}$ that minimizes the \textit{interim} expected cost:
\begin{equation}
    V_{A}(g_{D},h_{A}, r;\sigma)\equiv \sum\nolimits_{s,b}  u_{A}(s,b)h_{A}(s|r) \mu_{\sigma}(b|r).
\end{equation}
We let $\mathcal{BR}^{\sigma}[g_{D}]\equiv\{h_{A}|h_{A}\in\argmax\nolimits_{h_{A}} U_{A}(g_{D}, h_{A};\sigma)\}$ be the set of the attacker's ex-ante best responses to $g_{D}$ given $\sigma$, and let $\mathcal{BR}^{\sigma}_{\Gamma}[g_{D}]\equiv\{h_{A}|h_{A}\in\argmax\nolimits_{h_{A}} V_{A}(g_{D}, h_{A}, r;\sigma) \forall r\in\Gamma\}$ be the set of the attacker's interim best responses to $g_{D}$ given $\sigma$.

The interactions between the defender and the attacker can be modeled as a Bayesian game.
\textit{Bayesian Nash equilibrium} and \textit{perfect Bayesian Nash equilibrium} are two common equilibrium solution concepts for Bayesian games.

\begin{definition}[$\sigma$-Bayesian Nash Equilibrium ($\sigma$-BNE)]
    A profile $<g^{*}_{D}, h^{*}_{A}\}>$ is a \textup{$\sigma$-Bayesian Nash Equilibrium } if 
    \[
    \begin{aligned}
        g^{*}_{D}\in \argmax\nolimits_{g_{D}} U_{D}(g_{D}, h^{*}_{A};\theta) \textup{ and } h^{*}_{A}\in \mathcal{BR}^{\sigma}[g^{*}_{D}].
    \end{aligned}
    \]
\end{definition}

\begin{definition}[$\sigma$-Perfect Bayesian Nash Equilibrium ($\sigma$-PBNE)]
    A profile $<g^{*}_{D}, h^{*}_{D}>$ with posterior belief $\mu_{\sigma}$ is a \textup{$\sigma$-Perfect Bayesian Nash Equilibrium} if (i) $\mu_{\sigma}$ is updated according to Bayes' law, and (ii)
    \[
    \begin{aligned}
        g^{*}_{D}\in \argmax\nolimits_{g_{D}} U_{D}(g_{D}, h^{*}_{A};\theta) \textup{ and } h^{*}_{A}\in \mathcal{BR}^{\sigma}_{\Gamma}[g^{*}_{D}].
    \end{aligned}
    \]
\end{definition}
Since every PBNE is a BNE \cite{mas1995microeconomic}, it is not hard to obtain $\mathcal{BR}^{\sigma}_{\Gamma}[g_{D}]\subseteq \mathcal{BR}^{\sigma}[g_{D}]$, for any $g_{D}$.

\subsection{BNE: General-Sum GAN}
We train the BNE strategies using a GAN-like fashion termed general-sum GAN.
In particular, the defender's strategy is represented by a neural network \textit{generator} $G_{\lambda_D}(b,\nu)$ with parameter $\lambda_{D}$, which takes the true membership vector $b$ and an auxiliary vector $\nu$ as inputs and outputs a noise vector $\delta$ to perturb the summary statistics.
Here, we assume the vector $\nu$ with dimension $q$ has each coordinate uniform in $[0,1]$; $\tilde{\nu}\sim \mathcal{U}$
The attacker's strategy is represented by a neural network \textit{discriminator} $H_{\lambda_{A}}(r)$ with parameter $\lambda_{A}$, which takes the $r=\mathtt{R}(x+G_{\lambda_{A}}(b,\nu))$ as input and outputs an inference result $s$.
For simplicity, let $r(G_{\lambda_{A}}(b,\nu))=\mathtt{R}(x+G_{\lambda_{A}}(b,\nu))$
Define
\begin{align*}
    &\begin{aligned}
        U_{D}\left(G_{\lambda_{D}}, H_{\lambda_{A}}\right)\equiv \mathbb{E}^{\tilde{\nu}\sim \mathcal{U}}_{\tilde{b}\sim q}\left[\|G_{\lambda_{D}}(\tilde{b},\tilde{\nu})\|\right] + \kappa \mathbb{E}^{\tilde{\nu}\sim \mathcal{U}}_{\tilde{b}\sim q}\left[v\left(H_{\lambda_{A}}\left(r\left(G_{\lambda_{D}}(\tilde{b}, \tilde{\nu}) \right)\right), \tilde{b}\right)\right],
    \end{aligned}\\
    & \begin{aligned}
        &U^{\sigma}_{A}(G^{*}_{\lambda^{*}_{D}}, H^{*}_{\lambda^{*}_{A}})\\
        &\equiv \mathbb{E}^{\tilde{\nu}\sim \mathcal{U}}_{\tilde{b}\sim \sigma}\left[-v\left(H_{\lambda_{A}}\left(r\left(G_{\lambda_{D}}(\tilde{b}, \tilde{\nu}) \right)\right), \tilde{b}\right)\right]+ \gamma \mathbb{E}^{\tilde{\nu}\sim \mathcal{U}}_{\tilde{b}\sim \sigma}\left[c_{A}\left(H_{\lambda_{A}}\left(r\left(G_{\lambda_{D}}(\tilde{b}, \tilde{\nu}) \right)\right)\right)\right].
    \end{aligned}
\end{align*}
Thus, $G_{\lambda_{D}}$ and $H_{\lambda_{A}}$ play the following game:
\begin{equation}\label{eq:NN_BNE}
    \begin{aligned}
       &G^{*}_{\lambda^{*}_{D}}\in\arg\min\nolimits_{G_{\lambda_{D}}} U_{D}\left(G_{\lambda_{D}}, H^{*}_{\lambda^{*}_{A}}\right), 
       H^{*}_{\lambda^{*}_{A}}\in \arg\min\nolimits_{H_{\lambda_{A}}}U^{\sigma}_{A}(G^{*}_{\lambda^{*}_{D}}, H_{\lambda_{A}}).
    \end{aligned}
\end{equation}
The output of the neural network $H_{\lambda_{A}}$ is a vector of real values in $(0,1)$ due to sigmoid activation, which serves as a probabilistic estimate based on learned features of the noise-perturbed observation.
This output should be distinguished from a mixed-strategy probability of choosing $s_{k}=1$ for $k\in U$, derived not solely from observation data but also from considerations of minimizing the attacker's expected costs.
Let $p=(p_{k})_{k\in U}$ with each $p_{k}\in(0,1)$ denote the output of $H_{\lambda_{A}}$.
To approximate the original cost functions, we use the proxies for $v$ and $c_{A}$ given by $v(p,b)=\sum\nolimits_{k\in U}b_{k}p_{k} \textup{ and } c_{A}(p) = \sum\nolimits_{k\in U}p_{k}$, which are differentiable functions of each $p_{k}$.
Due to the linear transformations and the sigmoid function, each $p_{k}$ is a differentiable function of the neural networks' parameters.
Thus, the proxy cost functions are also differentiable with respect to these parameters.

\section{The Ordering of Privacy Robustness}

In this section, we compare the optimal $\alpha$-LRT attacks and the BNE $\sigma$-Bayesian attacks regarding the worst-case privacy loss (WCPL) from the vNM defender's perspective.
The degree of privacy risk of a $g_{D}$ is measured in terms of the worst-case privacy loss (WCPL) perceived by the vNM defender.
Under the $\alpha$-LRT attack, the WCPL $L(g_{D}, \tau^{*}, \alpha)$ is attained when the attack is Neyman-Pearson optimal.
Under the $\sigma$-Bayesian attack, the WCPL is given by $ L^{\sigma}(g_{D})\equiv \max\nolimits_{h_{A}\in\mathcal{BR}^{\sigma}[g_{D}]}L(g_{D}, h_{A})$,
%
where $L(g_{D}, h_{A})\equiv \sum_{b,s}\int_{r} v(s,b)h_{A}(s|r)\rho_{D}(r|b)dr q(b)$,
in which $q$ is the objective prior distribution of the membership.
In addition to the defender's strategy $g_{D}$, the $\sigma$-Bayesian attacker's subjective prior also impacts the privacy loss.
Since we do not consider parameterization or reparameterization for the priors over $W$, we consider the uniform distribution over $W$ as the \textit{non-informative prior} according to Laplace \cite{fienberg2006did}.
We restrict attention to subjective priors classified according to their informativeness with respect to the true prior $q$.

\begin{definition}[Aligned and Misaligned $\sigma$]\label{def:informative}
    Fix any $g_{D}$. We say that $\sigma$ is \textup{(weakly) informative} if $L(g_{D}, h^{\sigma}_{A})\leq L(g_{D}, h^{q}_{A})$,
    for $h^{\sigma}_{A}\in \mathcal{BR}^{\sigma}[g_{D}]$ and $h^{q}_{A}\in \mathcal{BR}^{q}[g_{D}]$.
    We say that $\sigma$ is \textup{non-informative} if it is a \textup{uniform distribution} over $W$.
    We refer to $\sigma$ as \textup{aligned (with $q$)} if it is either informative or non-informative, and as \textup{misaligned} if it is neither informative nor non-informative.
    In addition, $\sigma$ is \textup{strictly informative} if the strict inequality replaces the weak inequality.
\end{definition}


Thus, a $\sigma$ is aligned if $\sigma$ leads to best responses such that the attacker's true ex-ante expected cost is no larger than the true prior $q$ does.
The following theorem summarizes our first main result.

\begin{restatable}{theorem}{TheoremBayesianLRT}\label{thm:Bayesian_better_LRT}
Fix any $g_{D}$ and $\alpha$. 
Suppose $\sigma\in\Delta(W)$ is an \textup{aligned prior}. Then, it holds
\[
L^{\sigma}(g_{D})\geq L(g_{D}, \tau^{*}, \alpha)\geq L(g_{D}, \tau^{(N)}(\tilde{r}),\alpha_{\tau^{(N)}(\tilde{r})})\geq L(g_{D}, \tau^{o},\alpha_{\tau^{o}}),
\]
where $\mathbb{E}[\alpha_{\tau^{(N)}(\tilde{r})}]=\alpha$. 
Suppose $\sigma\in\Delta(W)$ is \textup{strict aligned}. Then, $L^{\sigma}(g_{D})> L(g_{D}, \tau^{*}, \alpha)$.
\end{restatable}

Theorem \ref{thm:Bayesian_better_LRT} demonstrates an ordering of WCPL perceived by the vNM defender encountering $\sigma$-Bayesian, optimal $\alpha$-LRT, adaptive $\alpha$-LRT, and naive $\alpha$-LRT attackers.
The WCPL of the vNM defender induced by $\sigma$-Bayesian attack is the worst among these four types of attacks.
However, the ordering (specifically, $L^{\sigma}(g_{D})$ versus $L(g_{D}, \tau^{*}, \alpha)$) does not in general hold when the attacker's subjective prior is misaligned.
We establish a set of necessary and sufficient conditions for determining the order of $L^{\sigma}(g_{D})$ versus $L(g_{D}, \tau^{*}, \alpha)$ under the framework of Gaussian $g_{D}$, where the subjective prior is arbitrary.

Consider any simple binary hypothesis testing problem: $\widehat{H}_{0}$ versus $\widehat{H}_{1}$.
Let two normal distributions $\mathcal{N}(0,1)$ and $\mathcal{N}(\widehat{\mathtt{M}}_{j},1)$ over $\mathcal{Y}_{j}$, respectively, be the corresponding probability distributions of generating $y_{j}$ from $\mathcal{Y}_{j}\subseteq \mathbb{R}$ under $\widehat{H}_{0}$ and $\widehat{H}_{1}$, for all $j=1,\dots,m$ with $1\leq m< \infty$.
By $\rho_{j}(\cdot|\widehat{H}_{i})\in\Delta(\mathcal{Y}_{j})$ we denote the associated density function under $\widehat{H}_{i}$ for $i\in\{0,1\}$.
For any $y=(y_{j})_{j\in Q}$, the LR statistics is given by $\mathcal{L}(y)\equiv \sum\nolimits_{j\in Q}\log\left(\rho_{j}(y_{j}|\widehat{H}_{0})/\rho_{j}(y_{j}|\widehat{H}_{1})\right)$.
For any Type-I error rate $\widehat{\alpha}$, let $\widehat{\beta}$ denote the minimum Type-II error rate that can be obtained for this hypothesis testing problem.
\begin{lemma}\label{lemma:normal_relationship}
    Let $\mathcal{F}\left(\alpha, \beta\right)\equiv\frac{\left(z_{\alpha}+z_{\beta}\right)^{2}\overline{\textup{V}}}{4\overline{\textup{M}}^{2}}$, where $z_{a}$ is the $100(1-a)$th percentile of the standard normal distribution, $\overline{\textup{M}}=\frac{1}{2}\sum_{j\in Q}\widehat{\mathtt{M}}^{2}_{j}$, and $\overline{\textup{V}}=\sum_{j\in Q}\widehat{\mathtt{M}}^{2}_{j}$.
    Then, the following holds. \textup{(i)} $\mathcal{F}\left(\widehat{\alpha}, \widehat{\beta}\right)=m$. 
    \textup{(ii)} Fix any $\hat{\alpha}$, as $m$ increases (resp. decreases), $\widehat{\beta}$ decreases (resp. increases).
\end{lemma}

\paragraph{Gaussian Mechanisms.}
Define $g_{D}(\delta|b)=\prod_{j\in Q} g^{j}_{D}(\delta_{j}|b)$, where each $g^{j}_{D}(\cdot|b)$ is the density function of the Gaussian distribution $\mathcal{N}(\mathtt{M}^{j}_{b}, \mathtt{V}^{j})$, where $\mathtt{M}_{b}$ is the mean and $\mathtt{V}_{b}$ is the variance, for $b\in U$, $j\in Q$.
Let $y=x+\delta=(x_{j}+\delta_{j})_{j\in Q}\in \mathcal{Y}\equiv\prod_{j\in Q}\mathcal{Y}_{j}$, where each $y_{j}=x_{j} +\delta_{j}\in \mathcal{Y}_{j}$.
Thus, the random variable $\tilde{r}=\mathtt{R}(\tilde{y})$ is the output of the post-processing $\mathtt{R}(\cdot)$ of the random variable $\tilde{y}$.
Suppose $\mathtt{R}$ is the identity function. That is, the attacker observes the un-clipped noisy observation $y$.
Let $\rho_{D}(\cdot|b)\in\Delta(\mathcal{Y})$ be the resulting conditional probability of the observation.
Let $b^{k}_{[0]}$ and $b^{k}_{[1]}$ be two \textit{adjacent} membership vectors differing only in individual $k$'s $b_{k}$: $b_{k}=0$ in $b^{k}_{[0]}$ and $b_{k}=1$ in $b^{k}_{[1]}$.
Given $Q$, define $\mu^{\sigma}_{0|1}[|Q|]\equiv\max_{k\in U}\sum_{s_{-k}}\int_{y}\mu_{\sigma}(s_{k}=0|y)\rho_{D}(y|b^{k}_{[1]})d y$ as the maximum probability of $s_{k}=0$ conditioning on any individual $b_{k}=1$ induced by the posterior belief $\mu_{\sigma}$ and $g_{D}$, where the maximum is over all individuals.

\begin{theorem}\label{thm:Gaussian_comparison}
Let $g_{D}$ be a Gaussian mechanism defined above with each $g^{j}_{D}(\cdot|b)\in\Delta(\mathcal{Y}_{j})$ as the density function of $\mathcal{N}(\mathtt{M}^{j}_{b}, \mathtt{V}^{j})$ given any $b\in W$.
Suppose $\mathtt{V}^{j}=\left(\frac{m}{K^{\dagger}\widehat{\mathtt{M}}_{j}}\right)^{2}$ for all $j\in Q$, where $1\leq K^{\dagger} \leq K$ is the minimum number of individuals involved in $B$.
Suppose in addition $\max_{b,b'}|\mathtt{M}^{j}_{b}-\mathtt{M}^{j}_{b'}|\leq \frac{m}{K^{\dagger}}$, where the maximum is over all adjacent membership vectors.
%
Then, for any $Q$ with $|Q|=m\geq 1$ and any arbitrary subjective prior $\sigma$, if $\mathcal{F}\left(\alpha, \mu^{\sigma}_{0|1}[m]\right)\geq m$, then $L(g_{D}, \tau^{*},\alpha)\leq L^{\sigma}(g_{D})$.
%
\end{theorem}

Theorem \ref{thm:Gaussian_comparison} provides a sufficient condition when $\sigma$-Bayesian outperforms $\alpha$-LRT under the Gaussian mechanisms when $\sigma$ is any arbitrary subjective prior, which is independent of the true prior $q$.
Given any $\alpha$, let $1-\beta^{\alpha}_{|Q|}$ denote the power of the $\alpha$-LRT attack and let $m^{\alpha}=\mathcal{F}\left(\alpha, \mu^{\sigma}_{0|1}[|Q|]\right)$ denote the number of SNVs used to calculate the summary statistics so that $1-\beta^{\alpha}_{|Q|}=1-\mu^{\sigma}_{0|1}[|Q|]$ (i.e., the power of the $\alpha$-LRT attack coincides with the worst-case TPR induced by the $\sigma$-Bayesian attack).
Since $m$ increases as $\hat{\beta}$ decreases by Lemma \ref{lemma:normal_relationship}, when $m^{\alpha}\geq |Q|$, the actual $1-\beta^{\alpha}_{|Q|}\leq 1-\mu^{\sigma}_{0|1}[|Q|]$.
Then, by Proposition \ref{thm:attacker_mirror_mu} in Appendix \ref{sec:app_proof_thm1}, the lowest TPR that can be obtained by the $\sigma$-Bayesian attacker is greater than the best power of the $\alpha$-LRT.
Therefore, $L(g_{D}, \tau^{*},\alpha)\leq Z(g_{D},\sigma)$.
Please refer to Appendix \ref{sec:app_GDM} for more information about the Gaussian mechanisms and a necessary and sufficient condition that depends on the true prior.

\section{Experiments}


Our experiments use a dataset of $800$ individuals with $5000$ SNVs of each individual on Chromosome $10$.
The dataset was provided by the 2016 iDASH Workshop on Privacy and Security \cite{tang2016idash}, derived from the 1000 Genomes Project \cite{10002015global}. 
Since the sigmoid function is used as the activation function, the output of the Bayesian attacker's neural network $H_{\lambda_{A}}$ is a vector with each element taking a value between $0$ and $1$, representing the probabilistic confidence of selecting each $s_{k}=1$.
The output of the defender's neural network $G_{\lambda_{D}}$ is the noise term set within the range $[-0.5,0.5]$.
We use the ROC (Receiver Operating Characteristic) Curve of the attacker's $H_{\lambda_{A}}$ to measure the strength of privacy protection by varying the thresholds to turn the confidence output of $H_{\lambda_{A}}$ to binary values $s_{k}\in\{0,1\}$ for $k\in U$, where a lower (resp. higher) AUC (Area Under the ROC Curve) indicates greater (resp. reduced) privacy effectiveness of $G_{\lambda_{D}}$.
We benchmark our Bayesian game-theoretic models with three baseline frameworks: defense against fixed-threshold attacks, defense against adaptive attacks, and defense using pure differential privacy.
The experiments were conducted using an NVIDIA A40 48G GPU. PyTorch was used as the deep learning framework.
Experiment statistical significance, neural network configurations, and hyperparameters are described in Appendix \ref{sec:app_Hyperparameters}.
To make (the approximation of) the objective function in (\ref{eq:nariveLRT}) differentiable with respect to the parameters of the generator $G_{\lambda_{D}}$, we use proxies.
In addition to using proxy $v_{D}$ for the defender's privacy loss function, we use the sigmoid function to approximate the threshold-based rejection rule of LRT.
In particular, $\mathbf{1}\{\ell(d_{k}, x)\leq\tau\}$ is approximated by $1/(1+\exp(-(\tau-\ell(d_{k}, x))))$, where $\ell(d_{k}, x)$ is the log-likelihood statistics.
Similarly, we use the sigmoid function to approximate $\mathbf{1}\{\ell(d_{k}, r)\leq \tau^{(N)}(r)\}$. 
The \textit{fixed- and adaptive-threshold LRT defenders} optimally choose $g_{D}$ by solving (\ref{eq:nariveLRT}) and (\ref{eq:adaptiveLRT}), respectively.

\subsection{Bayesian Defender vs. Three Types of Attackers}

We first compare the performance of the Bayesian defender under our Bayesian attacks, and two baseline attack models of the fixed-threshold and the adaptive-threshold attackers \cite{sankararaman2009genomic,shringarpure2015privacy,venkatesaramani2021defending,venkatesaramani2023enabling}.
In the experiments, we consider the true prior $q$ as the uniform distribution and assume that the Bayesian attacker uses $\sigma=q$ and $0<\gamma<1$.
In addition, we set $\kappa_{j}=\kappa$ for all $j\in Q$.
Figures \ref{fig:fig1}-\ref{fig:fig3} show the performance of the Bayesian defense against the Bayesian attacker, the fixed-threshold LRT attacker, and the adaptive-threshold LRT attacker for different values of the defender's parameter $\kappa$ that captures the defender's different preferences over privacy-utility trade-offs.
The experimental results support our theoretical analysis that (a) the Bayesian attacker is stronger than the existing LRT-based attacks, and (b) the Bayesian defense is effective against both fixed- and adaptive-threshold LRT attackers. 
In addition, the privacy of the defense decreases (resp. increases) as $\kappa$ increases (resp. decreases), as we would expect, since $\kappa$ captures the tradeoff between privacy and utility.

\subsection{Bayesian Attacker vs. Three Types of Defenders}

Next, we compare the performances of the Bayesian attacker under our Bayesian defense, and two state-of-the-art baseline defender models \cite{venkatesaramani2021defending,venkatesaramani2023enabling} that optimally respond to the fixed-threshold and the adaptive-threshold attackers, respectively. 
Figure \ref{fig:fig4} ($\kappa=1.5$) showcases the performances of the Bayesian attacker against different defenders, namely Bayesian, fixed-threshold LRT, and adaptive-threshold LRT defenders.
Here, the LRT defenders execute defense mechanisms that are optimally designed against their respective LRT attackers.
Notably, the Bayesian defense provides superior privacy protection against the Bayesian attacker, while the adaptive-threshold LRT defender outperforms its fixed-threshold counterpart under the Bayesian attacks.

\subsection{Bayesian Defender vs. Differential Privacy}

Finally, we compare our Bayesian defense mechanism with a conventional $\epsilon$-DP mechanism (see Appendix \ref{sec:app_differential_privacy} for more information).
In particular, we compare the performance of these two defense mechanisms under the Bayesian attack when \textit{(i)} both mechanisms cause the defender the same expected utility loss, but \textit{(ii)} the $\epsilon$-DP mechanism does not take into consideration the trade-off $\vec{\kappa}=(\kappa_{j})_{j\neq}$, where $\kappa_{j}$ is not the same for all $j\in Q$.
In the experiments, we consider the true prior $q$ as the uniform distribution and assume that the Bayesian attacker uses $\sigma=q$ and $0<\gamma<1$.
First, Figure \ref{fig:fig5} demonstrates the performances of the Bayesian, fixed-threshold, and adaptive-threshold attackers under $\epsilon$-DP defense where $\epsilon=600$.
The choice of such a large value of $\epsilon$ is explained in Appendix \ref{sec:app_differential_privacy}).
Similar to the scenarios under the Bayesian defense, the Bayesian attacker outperforms the LRT attackers under the $\epsilon$-DP.
Figure \ref{fig:fig6} illustrates an example under the Bayesian attack where the non-strategic DP defense results in significant privacy loss compared to the Bayesian defense, despite both defense mechanisms causing the same $\kappa$-weighted utility loss for the defender.
Specifically, we analyze a trade-off parameter vector $\kappa=(\kappa_{j})_{j\in Q}$ where $\kappa_{j}=0$ for the $90\%$ of $5000$ SNVs and $\kappa_{j}=50$ for the remaining $10\%$.
Under these conditions, the Bayesian attacker achieves an AUC of $0.53$ against the Bayesian defender and an AUC of $0.91$ against the $\epsilon$-DP defender.

\begin{figure}[htb]
    \centering
    \begin{subfigure}[b]{0.32\textwidth}
        \includegraphics[width=\textwidth]{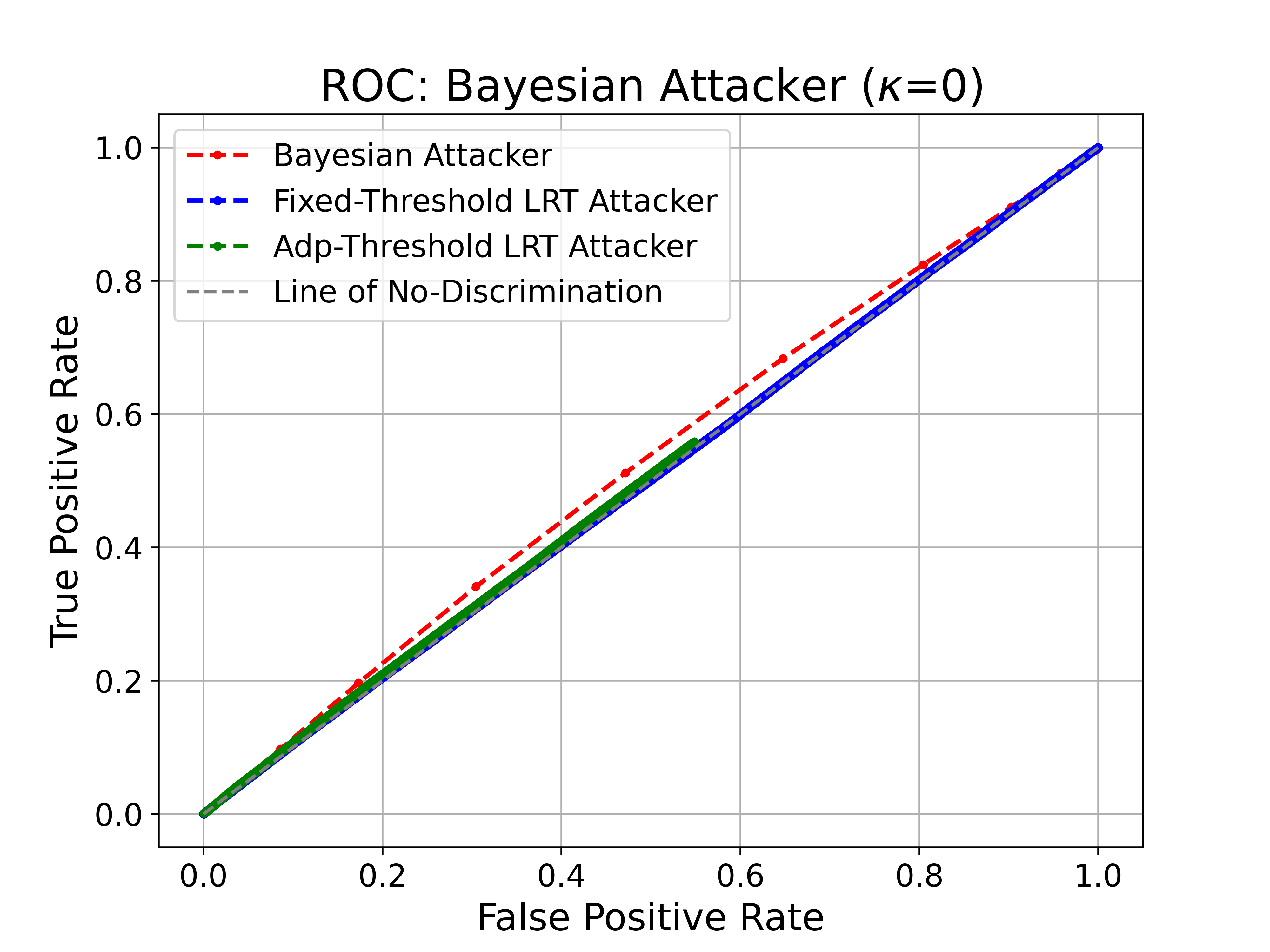}
        \caption{}
        \label{fig:fig1}
    \end{subfigure}
    \hfill 
    \begin{subfigure}[b]{0.32\textwidth}
        \includegraphics[width=\textwidth]{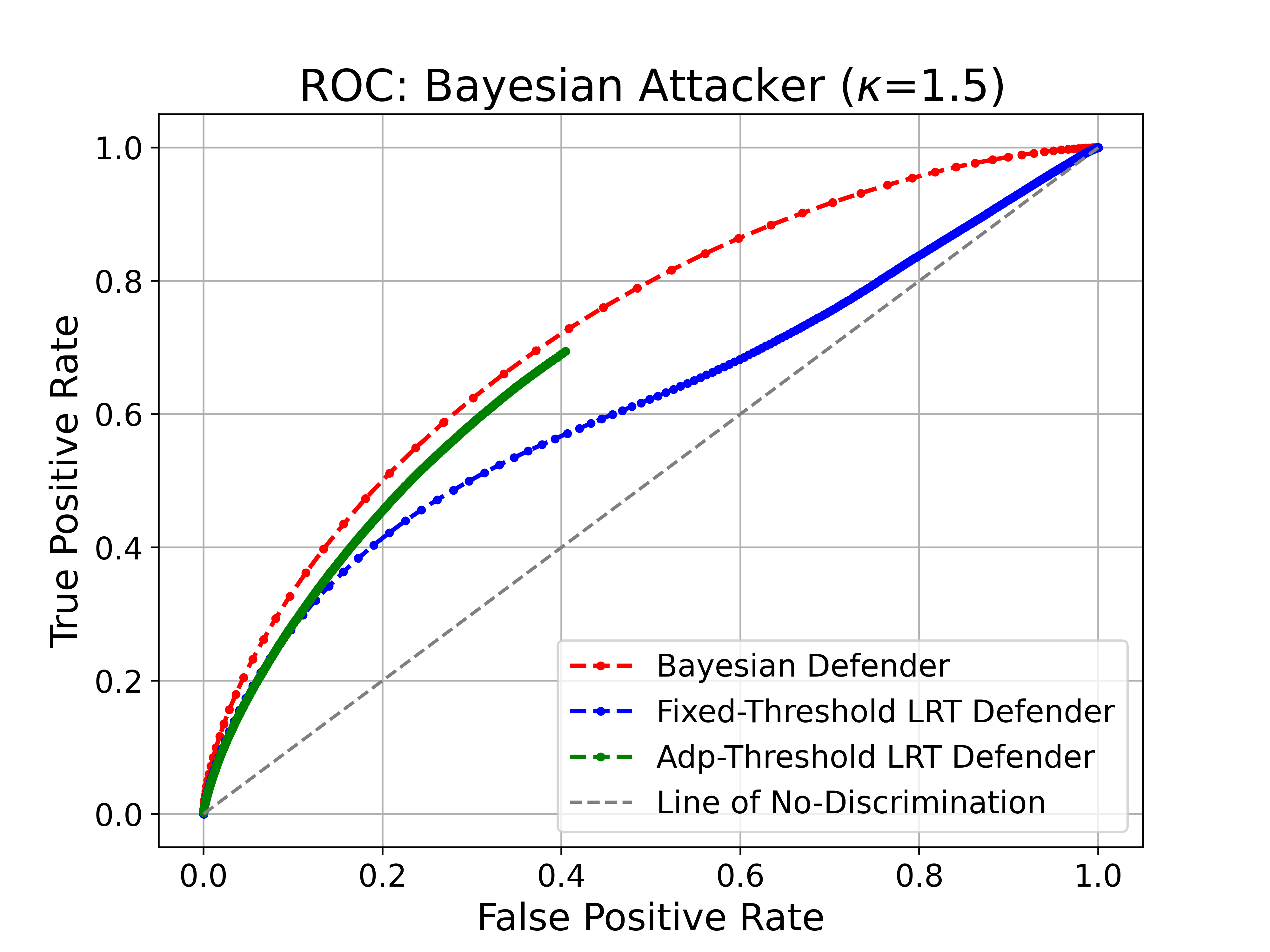}
        \caption{}
        \label{fig:fig2}
    \end{subfigure}
    \hfill
    \begin{subfigure}[b]{0.32\textwidth}
        \includegraphics[width=\textwidth]{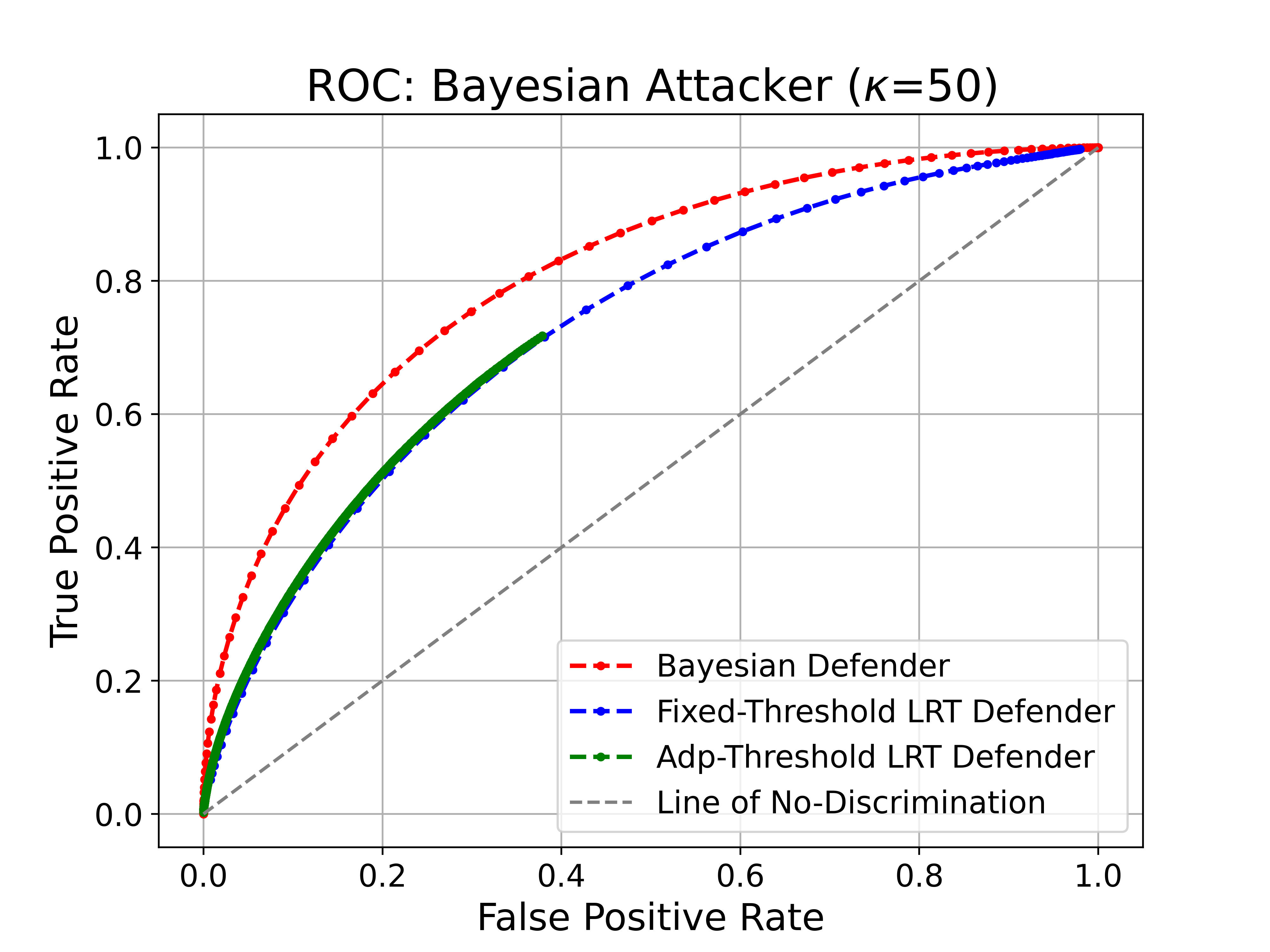}
        \caption{}
        \label{fig:fig3}
    \end{subfigure}
    
    \begin{subfigure}[b]{0.32\textwidth}
        \includegraphics[width=\textwidth]{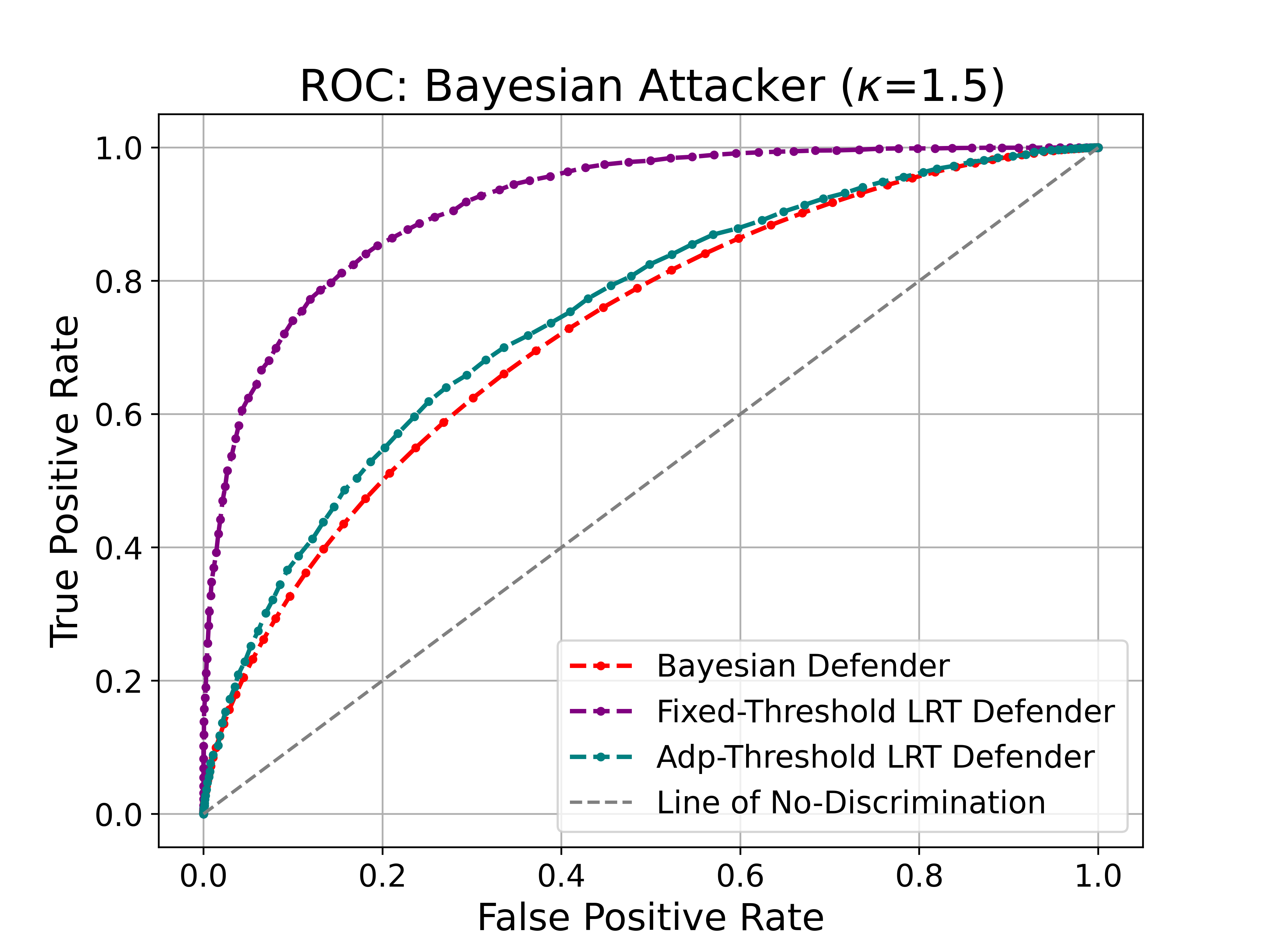}
        \caption{}
        \label{fig:fig4}
    \end{subfigure}
    \hfill
    \begin{subfigure}[b]{0.32\textwidth}
        \includegraphics[width=\textwidth]{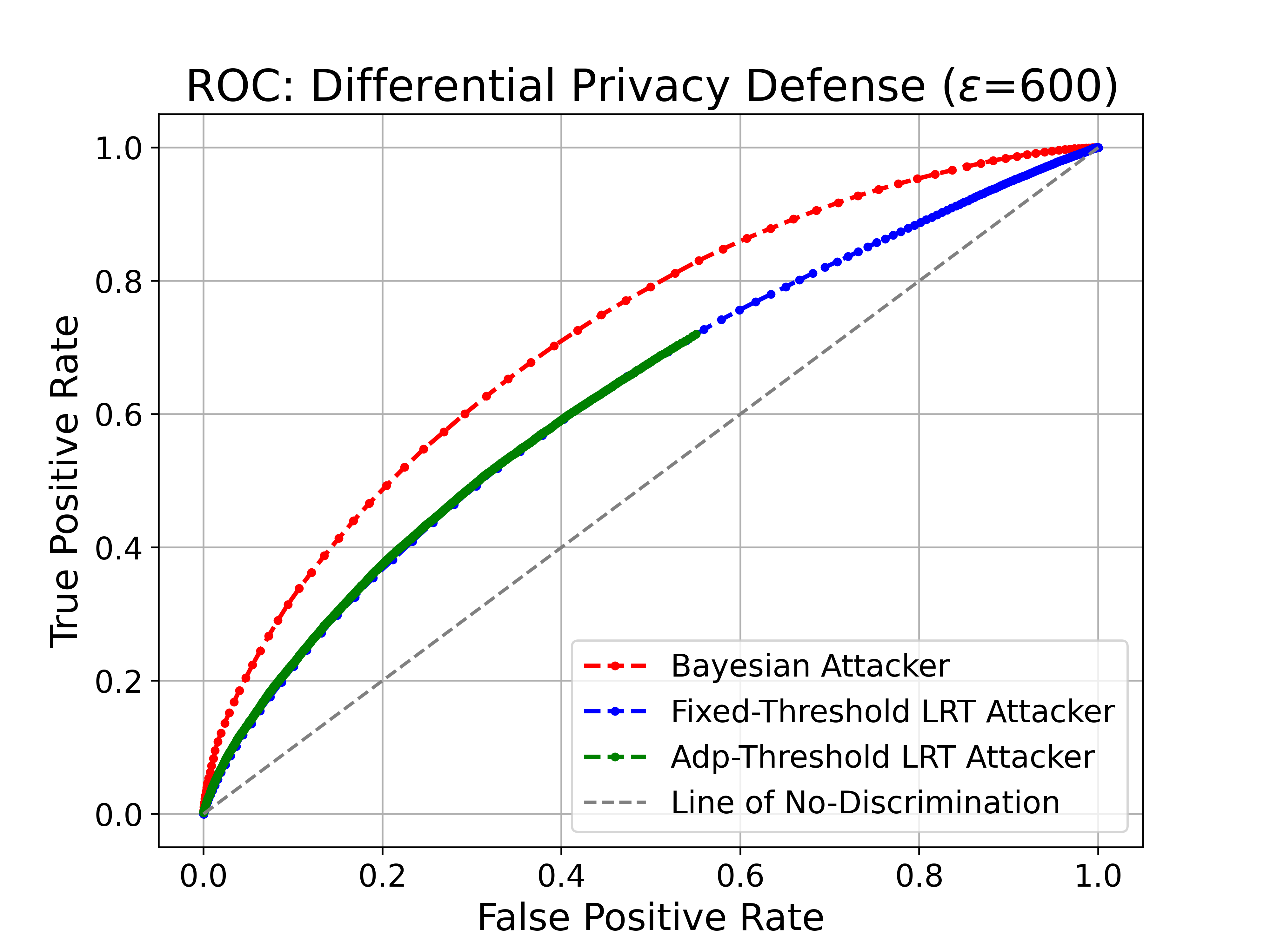}
        \caption{}
        \label{fig:fig5}
    \end{subfigure}
    \hfill
    \begin{subfigure}[b]{0.32\textwidth}
        \includegraphics[width=\textwidth]{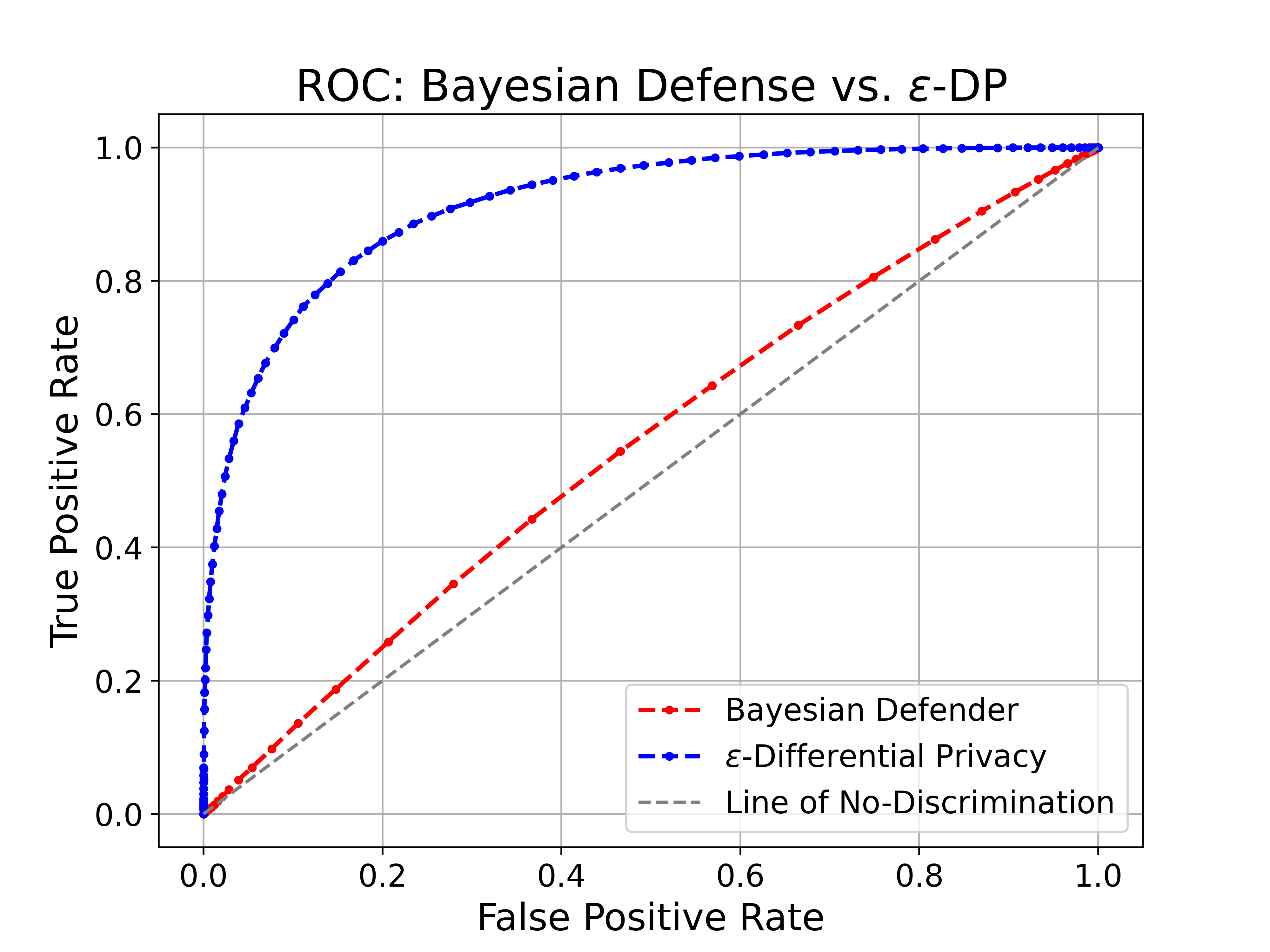}
        \caption{}
        \label{fig:fig6}
    \end{subfigure}
    
    \caption{(a)-(c): Bayesian Defender with $\kappa=0, 1.5, 50$, respectively. (d): Bayesian attacker under different defenders with $\kappa=1.5$. (e): Different attacks under non-strategic DP with $\epsilon=600$. (f): Bayesian attack under Bayesian defender and $\epsilon$-DP when $90\%$ of SNVs have $\kappa_{j}=0$ while $10\%$ have $\kappa_{j}=50$.}
    \label{fig:six_figures}
\end{figure}

\section{Conclusion}

This paper introduced a Bayesian game-theoretic framework to optimize the privacy-utility trade-off in genomic data sharing. Our theoretical analysis revealed that a Bayesian attacker, even with bounded rationality, poses a greater privacy risk to the defender than traditional LRT-based attackers. 
We compare Bayesian and LRT attacks under the Gaussian mechanism, highlighting conditions under which Bayesian attacks are more powerful. 
We also provided a practical method 
for approximating Bayes-Nash equilibria of this game, thereby computing a data sharing strategy that is explicitly trained to be robust to Bayesian attackers.
Our findings emphasize the importance of considering adaptive attackers in privacy risk assessments and offer a robust solution for privacy-preserving genomic data sharing.

\bibliographystyle{unsrt}  
\bibliography{references} 

\appendix

\section*{Appendix}

\section{Broader Impact, Limitations, and Future Work}

Our research investigates the privacy risks associated with genomic data-sharing services and proposes a Bayesian game-theoretic framework to optimize the trade-off between privacy and utility when sharing genomic summary statistics. By demonstrating that Bayesian attackers can induce greater privacy loss than conventional LRT-based attackers, our work underscores the pressing need for robust privacy-preserving mechanisms in genomic data sharing. This study contributes to the broader goal of developing secure methods for genomic data analysis, addressing critical societal concerns regarding individual privacy.
Our findings could help create guidelines for sharing genomic data responsibly, balancing the benefits of scientific collaboration with protecting individual privacy. This research aims to build public trust and enable responsible use of genomic data for scientific progress and better healthcare.

Our results are based on the assumption of linkage equilibrium. This requires the data-sharing mechanism to follow a prefiltering protocol that retains a subset of SNVs in linkage equilibrium, where each SNV is independent of the others. Although such a protocol is practically operationalizable, our robustness results cannot be generalized to real-world genomic data that exhibit linkage disequilibrium. By developing more sophisticated models that account for linkage disequilibrium, we may enhance the robustness and practical relevance of our findings. Additionally, advanced statistical techniques could further refine our privacy and utility estimates, improving the overall effectiveness of the proposed methods.

\section{More on Gaussian Defense Mechanisms}\label{sec:app_GDM}

Based on the sensitivity of $f$ (see the proof at Appendix \ref{app:proof_thm2} for details), 
Theorem \ref{thm:Gaussian_comparison} considers the worst-case bound of the powers of the LRT attack when the attacker knows the membership of every individual in the dataset except for a single individual. This bound is evaluated over all possible input membership vectors. Notably, the comparison in Theorem \ref{thm:Gaussian_comparison} is independent of the true prior distributions \( q = (q_k)_{k \in U} \) of the membership vectors and does not rely on specific true membership vectors forming the Beacon dataset.

When \(\mathcal{F}\left(\alpha, \mu^{\sigma}_{0|1}[m]\right) < m\), the lowest true positive rate (TPR) of the \(\sigma\)-Bayesian attacker is strictly smaller than the best power of the \(\alpha\)-LRT attacker. However, this does not guarantee that every TPR of the \(\sigma\)-Bayesian attacker is smaller than every power of the \(\alpha\)-LRT attacker across different Beacon datasets. Therefore, \(\mathcal{F}\left(\alpha, \mu^{\sigma}_{0|1}[m]\right) < m\) generally cannot imply that \(L(g_{D}, \tau^{*}, \alpha) > L^{\sigma}(g_{D})\).
Moreover, the condition \(\mathcal{F}\left(\alpha, \mu^{\sigma}_{0|1}[m]\right) \geq m\) is not necessary. That is, \(L(g_{D}, \tau^{*}, \alpha) \leq L^{\sigma}(g_{D})\) does not imply \(\mathcal{F}\left(\alpha, \mu^{\sigma}_{0|1}[m]\right) \geq m\) for any arbitrary subjective prior \(\sigma\).
We can also conclude that the sufficient condition in Theorem \ref{thm:Gaussian_comparison} is not applied only to aligned subjective priors.
The following corollary directly follows Theorem \ref{thm:Gaussian_comparison}.

\begin{corollary}
Given a Gaussian mechanism $g_{D}$ with $Q$, if the number of SNVs of the Beacon dataset satisfies $|Q|\leq \mathcal{F}\left(\alpha, \mu^{\sigma}_{0|1}[|Q|]\right)$, then the mechanism $g_{D}$ that is optimal to the $\sigma$-Bayesian attacks with any arbitrary $\sigma$ is guaranteed to be robust to any optimal $\alpha$-LRT attacks.
\end{corollary}

In this section, we relax Theorem  \ref{thm:Gaussian_comparison} and study the comparison between the Bayesian attacks with arbitrary subjective priors and the optimal LRT attacks without considering the worst-case bound of the powers of the LRT attacks.
Suppose in addition that the number of individuals involved in the Beacon dataset is fixed to be $0<n< K$. 
For ease of exposition, we consider the noises added to all SNVs to be iid. 
Consider a Gaussian mechanism $g_{D}(\delta|b)=\prod_{j\in Q} g^{j}_{D}(\delta_{j}|b)$, where each $g^{j}_{D}(\cdot|b)$ is the density function of $\mathcal{N}(\mathtt{M}_{b}, \mathtt{V})$.
Let two adjacent membership vectors $b^{[k]}_{0}$ and $b^{[k]}_{1}$ differing in individual $k$'s $b_k$, where $b^{[k]}_{0}$ has $b_{k}=0$ and $b^{[k]}_{1}$ has $b_{k}=1$.
Define two hypotheses: $H^{[k]}_{0}$: the true membership is $b^{[k]}_{0}$ vs. $H^{[k]}_{1}$: the true membership is $b^{[k]}_{1}$.
For any $k\in U$, it is straightforward to see that each $\tilde{y}_{j}=\tilde{x}_{j} + \tilde{\delta}_{j}$ is a Gaussian random variable.
That is, 
$\tilde{y}_{j}\sim \mathcal{N}(\mathtt{M}_{0}+x^{0}_{j}, \mathtt{V})$ under $H^{[k]}_{0}$ and $\tilde{y}_{j}\sim \mathcal{N}(\mathtt{M}_{1}+x^{1}_{j}, \mathtt{V})$ under $H^{[k]}_{1}$, where $\mathtt{M}_{i}= \mathtt{M}_{b^{[k]}_{i}}$ and $x^{i}_{j} = f(b^{[k]}_{i},d_{i})$ is the unperturbed summary statistics given $b^{[k]}_{i}$ with $d_{i}$, for $i\in\{0,1\}$.
Then, given any $(b^{[k]}_{0}, b^{[k]}_{1})$, the power of the optimal $\alpha$-LRT performed upon the observation $y_{j}$ for all $j\in Q$ can be obtained as
\[
\begin{aligned}
    T(\mathcal{N}(\mathtt{M}_{0}+x^{0}_{j}, \mathtt{V}), \mathcal{N}(\mathtt{M}_{1}+x^{1}_{j}, \mathtt{V}))(\alpha)=\Phi\left(\Phi^{-1}\left(1-\alpha\right) - \frac{|\mathtt{M}_{1} - \mathtt{M}_{0} + x^{1}_{j} - x^{0}_{j} |}{\sqrt{\mathtt{V}}}\right),
\end{aligned}
\]
where $\Phi$ is the cumulative distribution function (CDF) of the standard normal distribution.

Under the assumption of linkage equilibrium (i.e., each SNV is independent of the others), the power of the optimal $\alpha$-LRT performed upon $y=(y_j)_{j}$ can be obtained by the tensor product of $|Q|$ trade-off functions \cite{dong2021gaussian}.
In particular, the power can be represented by
\[
\begin{aligned}
    &T\left(\times_{j\in Q}\mathcal{N}(\mathtt{M}_{0}+x^{0}_{j}, \mathtt{V}),  \times_{j\in Q} \mathcal{N}(\mathtt{M}_{1}+x^{1}_{j}, \mathtt{V})\right)(\alpha)=T\left(\text{N}^{[k]}_{0}, \text{N}^{[k]}_{1}\right)(\alpha),
\end{aligned}
\]
where $\text{N}^{[k]}_{0} = \mathcal{N}(\mathtt{M}_{0}+x^{0}_{1}, \dots, \mathtt{M}_{0}+x^{0}_{|Q|}, \Sigma(\mathtt{V}))$ and $\text{N}^{[k]}_{1} = \mathcal{N}(\mathtt{M}_{1}+x^{1}_{1}, \dots, \mathtt{M}_{0}+x^{1}_{|Q|}, \Sigma(\mathtt{V}))$, in which $\Sigma(\mathtt{V})$ is a $|Q|\times|Q|$ diagonal matrix where each principal diagonal element is $\mathtt{V}$.
The Mahalanobis distance for the joint distributions is
\[
d_{\Sigma(V)}\left((\mathtt{M}_{0}+x^{0}_{1}, \dots, \mathtt{M}_{0}+x^{0}_{|Q|}), (\mathtt{M}_{1}+x^{1}_{1}, \dots, \mathtt{M}_{1}+x^{1}_{|Q|})\right)=\sqrt{\sum_{j\in Q}\frac{\left(\mathtt{M}_{1} - \mathtt{M}_{0} + x^{1}_{j} - x^{0}_{j} \right)^{2}}{\mathtt{V}}}.
\]
Therefore, we have
\[
\begin{aligned}
    T\left(\text{N}^{[k]}_{0}, \text{N}^{[k]}_{1}\right)(\alpha) =& \Phi\left(\Phi^{-1}\left(1-\alpha \right) -\sqrt{\sum_{j\in Q}\frac{\left(\mathtt{M}_{1} - \mathtt{M}_{0} + x^{1}_{j} - x^{0}_{j} \right)^{2}}{\mathtt{V}}} \right)\\
    =& T\left(\mathcal{N}(0,1), \mathcal{N}(\mathtt{M}_{eq}[b^{[k]}_{0}, b^{[k]}_{1}], 1)\right)(\alpha),
\end{aligned}
\]
where $\mathtt{M}_{eq}[b^{[k]}_{0}, b^{[k]}_{1}] = \sqrt{\sum_{j\in Q}\frac{\left(\mathtt{M}_{1} - \mathtt{M}_{0} + x^{1}_{j} - x^{0}_{j} \right)^{2}}{\mathtt{V}}}$, in which we show $[b^{[k]}_{0}, b^{[k]}_{1}]$ to indicate that the trade-off function is based on $b^{[k]}_{0}$ and $b^{[k]}_{1}$.

Let $b^{[k]}_{0}=(b_{k}=0, \hat{b}_{-k})$ and $b^{[k]}_{1}=(b_{k}=1, \hat{b}_{-k})$.
Define 
\[
\beta(\alpha,q)\equiv \sum\nolimits_{b_{k},\hat{b}_{-k}}T\left(\mathcal{N}(0,1), \mathcal{N}(\mathtt{M}_{eq}[b^{[k]}_{0}, b^{[k]}_{1}], 1)\right)(\alpha)q_{k}(b_{k}) q_{-k}(\hat{b}_{-k}),
\]
and
\[
\mu_{0|1}(\sigma, q)\equiv \sum_{b_{k},\hat{b}_{-k}}\sum_{s_{-k}}\int_{y}\mu_{\sigma}(s_{k}=0|y)\rho_{D}(y|b_{k}=1, \hat{b}_{-k})dyq_{k}(b_{k}) q_{-k}(\hat{b}_{-k}).
\]
In addition, define
\[
\Delta(\alpha, \sigma, q)\equiv \mu_{0|1}(\sigma,q)- \beta(\alpha,q). 
\]
The following corollary is straightforward.

\begin{corollary}\label{corollary:gaussian_prior}
    Let $g_{D}(\delta|b)=\prod_{j\in Q} g^{j}_{D}(\delta_{j}|b)$ be a Gaussian mechanism, where each $g^{j}_{D}(\cdot|b)$ is the density function of $\mathcal{N}(\mathtt{M}_{b}, \mathtt{V})$.
    Then, $L(g_{D}, \tau^{*},\alpha)\leq L^{\sigma}(g_{D})$ if and only if $\Delta(\alpha, \sigma, q)\geq 0$.
\end{corollary}

Corollary \ref{corollary:gaussian_prior} represents shows a condition for $L(g_{D}, \tau^{*}, \alpha)\leq L^{\sigma}(g_{D})$ when the Bayesian attacker's subjective prior $\sigma$ is arbitrary.
Here, $1-\beta(\alpha,q)$ is the expected power of the $\alpha$-LRT attacker perceived by the vNM defender, while $1-\mu_{0|1}$ is the expected posterior beliefs of $\{s_{k}=1\}_{k\in U}$.
Thus, $\Delta(\alpha, \sigma, q)\geq 0$ implies that the expected accuracy of inferring $\{s_{k}=1\}$ using the posterior beliefs is higher than the expected power of the $\sigma$-LRT.
By Proposition \ref{thm:attacker_mirror_mu}, we have that the Bayesian strategy that mirrors the posterior belief leads to the WCPL.
Therefore, given any $\rho_{D}$ and the true prior $q$, $\Delta(\alpha, \sigma, q)\geq 0$ is equivalent to $L(g_{D}, \tau^{*},\alpha)\leq L^{\sigma}(g_{D})$.
This condition is independent of the sensitivity of $f$ but depends on $g_{D}$ and the true prior $q$.

\subsection{LRT vNM Defender}

We use $g_{\mathtt{N}}$, $g_{\mathtt{Adp}}$, and $g_{\mathtt{Opt}}$ to denote the typical solutions to (\ref{eq:nariveLRT}), (\ref{eq:adaptiveLRT}) and (\ref{eq:optimalLRT}), respectively.
Suppose that all $g_{\mathtt{N}}$, $g_{\mathtt{Adp}}$, and $g_{\mathtt{Opt}}$ are Gaussian mechanisms.
We refer to the defender using $g_{\mathtt{N}}$, $g_{\mathtt{Adp}}$, and $g_{\mathtt{Opt}}$, respectively, as the naive, adaptive, and optimal \textit{LRT vNM defender}.
Then, the WCPL is captured by the power of the UMP test given a significant level $\alpha$.
Due to the Neyman-Pearson lemma, the WCPL is the power or the TPR of the optimal $\alpha$-LRT, $1-T[P^{k}_{0}(g_{D}), P^{k}_{1}(g_{D})](\alpha)$.

\begin{corollary}\label{corollary:LRT_defender}
    Fix any $g_{D}$ and $\alpha$.
    Let $\mathtt{TPR}(g_{D}, \sigma)$ denote the maximum TPR can be obtained by a $\sigma$-Bayesian attacker under $g_{D}$.
    Suppose that $g_{D}$ is chosen such that the WCPL is $1-T[P^{k}_{0}(g_{D}), P^{k}_{1}(g_{D})](\alpha)$.
    Then, the following hold.
    \begin{itemize}
        \item[(i)] If $\sigma$ is an informative or non-informative prior, then $\mathtt{TPR}(g_{D}, \sigma)\geq1-T[P^{k}_{0}(g_{D}), P^{k}_{1}(g_{D})](\alpha)$.
        \item[(ii)] Suppose that $g_{D}$ is Gaussian as described in Theorem \ref{thm:Gaussian_comparison}.
        If $\mathcal{F}\left(\alpha, \mu^{\sigma}_{0|1}[m]\right)\geq m$, then $\mathtt{TPR}(g_{D}, \sigma)\geq1-T[P^{k}_{0}(g_{D}), P^{k}_{1}(g_{D})](\alpha)$. If $\mathcal{F}\left(\alpha, \mu^{\sigma}_{0|1}[m]\right)< m$, then $\mathtt{TPR}(g_{D}, \sigma)< 1-T[P^{k}_{0}(g_{D}), P^{k}_{1}(g_{D})](\alpha)$.
    \end{itemize}
\end{corollary}

Part \textit{(i)} of Corollary \ref{corollary:LRT_defender} follows Theorem \ref{thm:Bayesian_better_LRT}.
In particular, from Theorem \ref{thm:Bayesian_better_LRT} we have $L^{\sigma}(g_{D})\geq L(g_{D}, \tau^{*}, \alpha)$ for aligned subjective priors.
Hence, $\mathtt{TPR}(g_{D}, \sigma)\geq1-T[P^{k}_{0}(g_{D}), P^{k}_{1}(g_{D})](\alpha)$.
Part \textit{(ii)} of Corollary \ref{corollary:LRT_defender} follows Theorem \ref{thm:Gaussian_comparison}.
If $\mathcal{F}\left(\alpha, \mu^{\sigma}_{0|1}[m]\right)\geq m$, Theorem \ref{thm:Gaussian_comparison} implies that $L(g_{D}, \tau^{*},\alpha)\leq L^{\sigma}(g_{D})$, which gives $\mathtt{TPR}(g_{D}, \sigma)\geq1-T[P^{k}_{0}(g_{D}), P^{k}_{1}(g_{D})](\alpha)$.
If $\mathcal{F}\left(\alpha, \mu^{\sigma}_{0|1}[m]\right)< m$, then $L(g_{D}, \tau^{*},\alpha)> L^{\sigma}(g_{D})$, which implies $\mathtt{TPR}(g_{D}, \sigma)< 1-T[P^{k}_{0}(g_{D}), P^{k}_{1}(g_{D})](\alpha)$.

\section{Differential Privacy}\label{sec:app_differential_privacy}

\paragraph{Standard Differential Privacy}
Differential privacy \cite{dwork2006calibrating,dwork2006differential}
is a widely used data privacy preservation technique based on probabilistic distinguishability. 
Formally, we say a randomized mechanism $F$ is $(\epsilon,\varrho)$-differentially private if for any two adjacent dataset $\mathtt{D}$ and $\mathtt{D}'$ differing in only one entry if holds that
\[
\textup{P}\left(F((\mathtt{D}'))\in \mathcal{F}\right)\leq e^{\epsilon}\textup{P}\left(F(\mathtt{D}')\in \mathcal{F}\right) + \varrho
\]
for any possible subset $\mathcal{F}$ of the image of the mechanism $F$.
The parameter $\epsilon$ is usually referred to as the \textit{privacy budget}, which is small but non-negligible. 
$(\epsilon,0)$-DP or $\epsilon$-DP is known as \textit{pure differential privacy}, while with a non-zero $\varrho>0$, $(\epsilon,\varrho)$-DP is viewed as \textit{approximate differential privacy}.

\paragraph{Sensitivity}
Define the \textit{sensitivity} of $f$ by
\[
\mathtt{sens}(f) \equiv \max_{b,b'}|f(b,d) - f(b', d')|,
\]
where the maximum is over all adjacent datasets $(b, d)$ and $(b', d')$ where $b$ and $b'$ differs only in a single individual with $d$ and $d'$ as the corresponding SNVs, respectively.
For a given SNV in a dataset with $B\subseteq U$, $d_{kj}$ is either $0$ or $1$.
Thus, the maximum possible difference between the averages over the columns that differ in one entry is $\frac{1}{|B|}$.
Let $1\leq K^{\dagger} \leq K$ be the minimum number of individuals involved in the Beacon dataset.
Hence, $\mathtt{sens}(f) = \frac{m}{K^{\dagger}}$.
Suppose we choose $g_{D}$ as a Laplace mechanism.
That is, $g_{D}(\cdot|b)$ is $\mathtt{Laplace}(0, \frac{ \mathtt{sens}(f)}{\epsilon})$, for all $b\in W$.
Then, it satisfies (pure) $\epsilon$-differential privacy if $\mathtt{R}$ is the identity function since the Laplace mechanism performs output perturbation \cite{dwork2006differential}.
Due to the post-processing property of the standard differential privacy, it is clear that the Laplace mechanism $g_{D}$ is also $\epsilon$-differentially private for any non-identity $\mathtt{R}$.

\paragraph{Choice of $\epsilon$}
The sensitivity of the summary statistics function $f(\cdot)$ has sensitivity $\frac{m}{K^{\dagger}}$ (see Appendix \ref{sec:app_differential_privacy}), where $m=|Q|$ and $1\leq K^{\dagger}\leq K$ is the minimum number of individuals in $U$ involved in the Beacon dataset.
In general, $m \gg K$. Hence, a small value of $\epsilon$ (e.g., on the order of one to $10$) leads to very large scalar for the Laplace distribution. Thus, we choose a relatively large value of $\epsilon$ (e.g., $\epsilon=600$) in our experiments.

\paragraph{Gaussian Differential Privacy }
Next, we consider the scenario when $g_{D}$ is a Gaussian mechanism described in Theorem \ref{thm:Gaussian_comparison}.
In particular, given any $b\in W$, $g^{j}_{D}(\cdot|b)\in\Delta(\mathcal{Y}_{j})$ is the density function of $\mathcal{N}(\mathtt{M}^{j}_{b}, \mathtt{V}^{j})$ for all $j\in Q$, where $\mathtt{V}^{j}=\left(\frac{m}{K^{\dagger}\widehat{\mathtt{M}}_{j}}\right)^{2}$ and $\max_{b,b'}|\mathtt{M}^{j}_{b}-\mathtt{M}^{j}_{b'}|\leq \frac{m}{K^{\dagger}}$, for all $j\in Q$.
By Lemma \ref{lemma:gpd}, we have
\[
    T\left[P_{b}(g^{j}_{D}), P_{b'}(g^{j}_{D})\right](\alpha)\geq T\left[\mathcal{N}(0,1), \mathcal{N}(\widehat{\mathtt{M}}_{j}, 1)\right],
\]
for all adjacent $b$ and $b'$.
Therefore, each $g^{j}_{D}$ satisfies $\widehat{\mathtt{M}}_{j}$-\textit{Gaussian differential privacy} ($\widehat{\mathtt{M}}_{j}$-GDP) \cite{dong2021gaussian}, for all $j\in Q$.
By Corollary 2.1 of \cite{dong2021gaussian}, this $\widehat{\mathtt{M}}_{j}$-GDP mechanism $g^{j}_{D}$ is also $(\epsilon_{j}, \varrho_{j}(\epsilon_{j}))$-DP for all $\epsilon_{j}\geq 0$ with
\[
\varrho_{j}(\epsilon_{j}) = \Phi\left(-\frac{\epsilon_{j}}{\widehat{\mathtt{M}}_{j}} + \frac{\widehat{\mathtt{M}}_{j}}{2}\right) - e^{\epsilon_{j}}\Phi\left(-\frac{\epsilon_{j}}{\widehat{\mathtt{M}}_{j}} - \frac{\widehat{\mathtt{M}}_{j}}{2}\right),
\]
where $\Phi$ is the cumulative distribution function (CDF) of the standard normal distribution.
Under the assumption of linkage equilibrium and the construct of $g_{D}(y|b)=\prod_{j\in Q}g^{j}_{D}(y_j|b)$, the Gaussian defense mechanism $g_{D}$ is $\overline{\mathtt{M}}$-GDP with $\overline{\mathtt{M}} = \sqrt{\sum\nolimits_{j\in Q}\widehat{\mathtt{M}}^{2}_{j}}$ \cite{dong2021gaussian} due to the composition property.

\section{Proof of Theorem \ref{thm:Bayesian_better_LRT}}\label{sec:app_proof_thm1}

We start by proving Proposition \ref{thm:attacker_mirror_mu}.
Define 
\[
Z(g_{D},\sigma)\equiv \sum\nolimits_{b,s}\int_{r}v(s,b)\mu_{\sigma}(s|r)\rho_{D}(r|b)dr q(b),
\]
which is independent of the $\sigma$-Bayesian attacker's strategy $h_{A}$ and the test conclusions of $\alpha$-LRT attacker.

\begin{restatable}{proposition}{TheoremAttMirrMu}\label{thm:attacker_mirror_mu}
For any $g_{D}$ and $\sigma$, $Z(g_{D},\sigma)=L^{\sigma}(g_{D})$.
\end{restatable}


Proposition \ref{thm:attacker_mirror_mu} demonstrates that the WCPL of a vNM defender under $\sigma$-Bayesian attacker can be fully characterized endogenously by $g_{D}$ and the attacker's subjective prior $\sigma$, and is independent of any BNE strategy $h_{D}\in \mathcal{BR}^{\sigma}[g_{D}]$.
Moreover, Proposition \ref{thm:attacker_mirror_mu} does not assume the informativeness of the subjective prior.
However, the WCPL ordering of Bayesian and LRT attacks shown in Theorem \ref{thm:Bayesian_better_LRT} requires the subjective priors to be informative or non-informative.
In general, when the subjective priors are misaligned, the comparison between $\sigma$-Bayesian and $\alpha$-LRT attacks needs to consider $\sigma$ and $\alpha$ in a case-by-case manner.
A straightforward way is to compare the corresponding WCPLs; i.e., $L^{\sigma}(g_{D})\leq (\textup{resp. }\geq) L(g_{D}, \tau^{*},\alpha)$ implies that the $\alpha$-LRT attack is (resp. not) worse than the $\sigma$-Bayesian attack for the vNM defender when $\sigma$ is misaligned.
We further compare $\sigma$-Bayesian and $\alpha$-LRT attacks when $\sigma$ is an arbitrary subjective priors in the next subsection when the defender adopts Gaussian defense mechanisms in a theoretically ideal manner.

\subsection{Poof of Proposition \ref{thm:attacker_mirror_mu}}\label{sec:app_proposition_1}

Given $\mu_{\sigma}$ (determined by $g_{D}$ and $\sigma$) and any $h_{A}$, let $\widehat{U}_{A}(h_{A}, b, r)\equiv \sum_{s}u_{A}(s,b)h_{A}(s|r)\mu_{\sigma}(b|r)$, which depends on the membership vector $b$ sampled by $\mu_{\sigma}$ but is independent of the samples $s$ drawn by $h_{A}$.
Define 
\[
S^{*}[b,r;g_{D}]\equiv\left\{h_{A}(\cdot|r)\middle|h_{A}(\cdot|r)\in\argmin\nolimits_{h'_{A}} \widehat{U}_{A}(h'_{A}, b, r)\right\},
\]
for all $b\in W$ with $\mu_{\sigma}(b|r)>0$, any $r\in\Gamma$,
where $S^{*}[b,r;g_{D}]$ depends on $g_{D}$ through $\mu_{\sigma}$.
We first show that there is a $h^{*}_{A}(\cdot|r)\in S^{*}[b,r;g_{D}]$ that assigns probability $1$ to $b$ (with $\mu_{\sigma}(b|r)>0$).
Suppose in contrast that $0\leq h^{*}_{A}(b|r)<1$.
Then, it holds that $\sum_{s: s\neq b}u_{A}(s,b)h_{A}(s|r)\mu_{\sigma}(b|r)>0$, which gives
\[
\begin{aligned}
    &\widehat{U}_{A}(h^{*}_{A}, b, r)= \sum\nolimits_{s}u_{A}(s,b)h_{A}(s|r)\mu_{\sigma}(b|r)\\
    &=\sum\nolimits_{s: s\neq b}u_{A}(s,b)h_{A}(s|r)\mu_{\sigma}(b|r)+ u_{A}(b,b)h_{A}(b|r)\mu_{\sigma}(b|r)>u_{A}(b,b)h_{A}(b|r)\mu_{\sigma}(b|r).
\end{aligned}
\]
Thus, $\widehat{U}_{A}(h^{*}_{A}, b, r)|_{h^{*}_{A}(b|r)\neq 1}>\hat{U}_{A}(h^{'}_{A}, b, r)|_{h^{'}_{A}(b|r)= 1}$, which contradicts to $h^{*}_{A}(\cdot|r)\in S^{*}[b,r;g_{D}]$.
Therefore, we have $u_{A}(b,b)\mu_{\sigma}(b|r)\leq \sum_{s}u_{A}(s,b)h_{A}(s|r)\mu_{\sigma}(b|r)$, for all $h_{A}(\cdot|r)$, $b\in W$, $r\in \Gamma$, where the equality holds when $h_{A}(\cdot|r)\in S^{*}[b,r;g_{D}]$.

Let $h^{\mu}_{A}:\Gamma\mapsto \Delta(W)$ mirror the posterior belief $\mu_{\sigma}$; i.e., $h^{\mu}_{A}(s|r)=\mu_{\sigma}(b|r)\mathbf{1}\{s=b\}$, for all $s,b\in W$, $r\in\Gamma$.
It is clear that $h^{\mu}_{A}(\cdot|r)\in S^{*}[b,r;g_{D}]$ for all $b\in W$.
Next, we show that if $h^{\mu}_{A}(s|r)$ is used by the $\sigma$-Bayesian attacker, it induces the WCPL for the vNM defender, which is captured by Proposition \ref{prop:mu_least_privacy}.

\begin{proposition}\label{prop:mu_least_privacy}
    Given any $g_{D}$ and $\sigma$, $L(g_{D}, h^{\mu}_{A},\sigma)\leq L(g_{D},h^{*}_{A},\sigma)$, for all $h^{*}_{A}\in\mathcal{BR}^{\sigma}[g_{D}]$. 
\end{proposition}

\begin{proof}

Define $\pi\equiv h_{A}\circ \rho_{D}:W\mapsto\Delta(W)$ by $\pi(s|b)= \sum_{r}h_{A}(s|r)\rho_{D}(r|b)$, for all $s,b\in W$.
Define the set 
\[
\Pi[g_{D}]\equiv\left\{\pi= h_{A}\circ \rho_{D}\middle|h_{A}:\Gamma\mapsto\Delta(W)\right\}.
\]
That is, $\Pi[g_{D}]$ is the set of all feasible probabilistic mappings from a true membership vector $b$ to an inference $s$, perceived by the defender.
We first establish the following lemma
regarding the informativeness of $g_{D}$ in the sense of Blackwell's ordering of informatinveness \cite{blackwell1951comparison,de2018blackwell}.

\begin{lemma}\label{lemma:attacker_preference_inclusive}
Fix any $\sigma\in\Delta(W)$.
        Given any two $g_{D}, g'_{D}$, 
       $\Pi[g_{D}]\subseteq\Pi[g'_{D}]$, if and only if, for any function $\zeta:W\times W\mapsto \mathbb{R}$,
       \[
       \sum\nolimits_{b,s}\zeta(s,b) \pi'(s|b)\sigma(b) \leq \sum\nolimits_{b,s}\zeta(s,b) \pi(s|b)\sigma(b),
       \]
       where $\pi\in\Pi[g_{D}]$ and $\pi'\in\Pi[g'_{D}]$.
    \end{lemma}

    \begin{proof}
    We start by showing the \textit{``only if"} part.
    Let $\Pi^{*}[g_{D}]\equiv\{\pi|\pi\in \argmin_{\pi\in \Pi[g_{D}]}\sum\nolimits_{b,s}\zeta(s,b) \pi'(s|b)\sigma(b) \}$.
    Since $\Pi[g_{D}]\subseteq \Pi[g'_{D}]$ and $\Pi^{*}[g_{D}]\subseteq \Pi[g_{D}]$, it must hold that $\Pi^{*}[g_{D}]\subseteq \Pi[g'_{D}]$.
    Hence, 
    \[
    \sum\nolimits_{b,s}\zeta(s,b) \pi'(s|b)\sigma(b)\leq \sum\nolimits_{b,s}\zeta(s,b) \pi(s|b)\sigma(b),
    \]
    for all $\pi\in\Pi[g_{D}]$ and $\pi'\in\Pi[g'_{D}]$.

    Next, we show the \textit{``if"} part.
    Suppose in contrast that $\Pi[g_{D}]\not\subseteq \Pi[g'_{D}]$.
     Then, there exists a $\pi\in \Pi[g_{D}]$ such that $\pi\not\in \Pi[g'_{D}]$. 
     Since the set $\Pi[\bar{g}_{D}]$ for every $\bar{g}_{D}:W\mapsto\Delta(D)$ is closed under convex combinations of its elements, it is convex.
     In addition, it is a continuous image of a compact set in the space of probability distributions.
     Hence, the set $\Pi[\bar{g}_{D}]$ is also compact. 
     The set $\Pi[\bar{g}_{D}]$ can be seen as a subset of $\mathbb{R}^{W\times W}$.
     Therefore, we can also perceive $\pi \in \mathbb{R}^{W\times W}\backslash \Pi[g'_{D}]$.

     Let $\pi^{\sigma}(s,b)\equiv \pi(s|b)\sigma(b)$ for all $s,b\in W$.
    With abuse of notation, let $\Pi[g''_{D}, \sigma]\equiv \{\pi^{\sigma}|\pi\in \Pi[g''_{D}]\}$.
    Then, the set $\Pi[g''_{D}, \sigma]$ is a subset of $\mathbb{R}^{W\times W}$.
    Thus, $\pi^{\sigma} \in \mathbb{R}^{W\times W}\backslash \Pi[g'_{D},\sigma]$.
    Let $\hat{\zeta}\in \mathbb{R}^{W\times W}$ represents the matrix form of the function $\zeta$.
Since $|W|=2^{K}$ with $K>1$, there exists a separating hyperplane orthogonal to $\hat{\zeta}$, which separates the set $\Pi[g'_{D}]$ from the point $\pi$, such that
\[
\begin{aligned}
    \sum\nolimits_{b,s}\zeta(s,b) \pi'(s|b) > \sum\nolimits_{b,s}\zeta(s,b) \pi(s|b),
\end{aligned}
\]
for all $\pi'\in \Pi[g'_{D}]$.
Then, the attacker with a non-informative (i.e., uniform prior) $\sigma$ obtains an ex-ante expected payoff using $h_{A}$ such that $\pi=h_{A}\circ g_{D}$ that is strictly better than any $h'_{A}$ such that $h'_{A}\circ \rho_{D}\in \Pi[g'_{D}]$.
Thus, we obtain a contradiction to $\sum\nolimits_{b,s}\zeta(s,b) \pi'(s|b)\sigma(b) \leq \sum\nolimits_{b,s}\zeta(s,b) \pi(s|b)\sigma(b)$ for all $\sigma\in\Delta(W)$.
\end{proof}

Next, we want to show that $\Pi[g_{D}]\subseteq \Pi[g'_{D}]$ is equivalent to $g'_{D}=\eta \circ g_{D}$ for some garbling $\eta: \Gamma \mapsto \Delta(\Gamma)$, which is another format of Blackwell's ordering of information structures \cite{blackwell1951comparison,de2018blackwell}.

\begin{lemma}\label{lemma:blackwell_garbling}
    For any two $g_{D}, g'_{D}$, $\Pi[g'_{D}]\subseteq \Pi[g_{D}]$ if and only if $g'_{D}=\eta \circ g_{D}$ for some garbling $\eta: \Gamma \mapsto \Delta(\Gamma)$.
\end{lemma}

\begin{proof}
    If $g'_{D}=\eta\circ g_{D}$, then there is a garbling $\hat{\eta}:\Gamma\mapsto\Delta(\Gamma)$ such that  $\rho'_{D}=\hat{\eta}\circ \rho_{D}$. Hence, $\pi'=\hat{\eta}\circ \pi$ for every $\pi'\in \Pi[g'_{D}]$ and $\pi\in \Pi[g_{D}]$. Then, from (1) and (2) of Theorem 1 in \cite{de2018blackwell}, we obtain $g'_{D}=\hat{\eta}\circ g_{D}$ is equivalent to $\Pi[g'_{D}]\subseteq \Pi[g_{D}]$.
\end{proof}

For simplicity, let $\mu^{r}_{\sigma}=\mu_{\sigma}(\cdot|r)$.
Since $h_{A}\in \mathcal{BR}^{\sigma}[g_{D}]$, there exists a randomized correspondence $Y$ such that $h_{A}(\cdot|r)=Y(\cdot|\mu^{r}_{\sigma})$ for all $r\in\Gamma$.
Then, from Blackwell's theorem \cite{blackwell1951comparison,de2018blackwell}, there exists a garbling $y:W\mapsto \Delta(W)$ such that $h_{A}=y\circ \mu_{\sigma}$.
Let $\hat{\rho}_{D}\equiv \mu_{\sigma}\circ \rho_{D}$ and let $\hat{\rho}'_{D}\equiv y\circ \hat{\rho}_{D}$.
In addition, let $\hat{g}_{D}$ and $\hat{g}'_{D}$, respectively, be corresponding to $\hat{\rho}_{D}$ and $\hat{\rho}'_{D}$.
Then, from Lemma \ref{lemma:blackwell_garbling}, we have $\Pi[\hat{g}'_{D}]\subseteq \Pi[\hat{g}_{D}]$.
In addition, Lemma \ref{lemma:attacker_preference_inclusive} implies that 
\[
\sum\nolimits_{b,s}\zeta(s,b) \hat{\pi}(s|b)\sigma(b) \leq \sum\nolimits_{b,s}\zeta(s,b) \hat{\pi}'(s|b)\sigma(b)
\]
for any $\sigma\in\Delta(W)$, any function $\zeta:W\times W\mapsto \mathbb{R}$, where $\hat{\pi}\in\Pi[\hat{g}_{D}]$ and $\hat{\pi}'\in \Pi[\hat{g}'_{D}]$.
If we take $\zeta(\cdot) = v(\cdot)$ and $\sigma(\cdot)=q(\cdot)$, then we have $L(\hat{g}_{D})\geq L(\hat{g}'_{D})$.
Therefore, $L(g_{D}, h^{\mu}_{A})\leq L(g_{D},h^{*}_{A})$ for all $h^{*}_{A}\in\mathcal{BR}^{\sigma}[g_{D}]$, which concludes the proof of Proposition \ref{prop:mu_least_privacy}.
\end{proof}

Next, we show that there is a $h^{*}_{A}\in\mathcal{BR}^{\sigma}[g_{D}]$ such that $h^{*}_{A}(s|r)=h^{\mu}_{A}(s|r)$ for all $s\in W$, $r\in\Gamma$.
Define $\widehat{U}^{\natural}(s, r)\equiv \sum_{b}u_{A}(s,b)\mu_{\sigma}(b|r)$, which depends on samples of $s\in W$ and $r\in\Gamma$.
Let 
\[
W^{\natural}[r]\equiv\left\{s\in W\middle| s\in\argmin\nolimits_{s'} \widehat{U}^{\natural}(s', r)\right\}.
\]
Let $\hat{s}\in W$ such that $\hat{h}_{A}(\hat{s}|r)=1$ for $\hat{h}_{A}\in S^{*}[b,r;g_{D}]$.
We want to show $\hat{s}\in W^{\natural}[r]$.
Suppose in contrast that $\hat{s}\not\in W^{\natural}[r]$.
Then, $\widehat{U}^{\natural}(s, r)< \widehat{U}^{\natural}(\hat{s}, r)$ for all $s\in W^{\natural}[r]$.
That is, $\sum_{b}u_{A}(s,b)\mu_{\sigma}(b|r)<\sum_{b}u_{A}(\hat{s},b)\mu_{\sigma}(b|r)$.
Since $\hat{h}_{A}\in S^{*}[b,r;g_{D}]$, we have $u_{A}(\hat{s}=b, b)\leq \sum_{s}\mu_{A}(s,b)h'_{A}(s|r)$ for all $h'_{A}(\cdot|r)$, including $h'_{A}(s|r)=1$ for any $s\in W$.
Since every $\mu_{\sigma}(\cdot)\geq0$, we have $u_{A}(\hat{s}=b, b)\mu_{\sigma}(b|r)\leq \mu_{A}(s,b)\mu_{\sigma}(b|r)$, for all $s,b\in W$.
Then, $\widehat{U}^{\natural}(\hat{s}, r)\leq \widehat{U}^{\natural}(s, r)$, contradicting to $\hat{s}\not\in W^{\natural}[r]$.
Therefore, $\hat{s}\in W^{\natural}[r]$.

Next, we show that for every $s^{*}\in W^{\natural}[r]$, there is a $b\in W$ with $\mu_{\sigma}(b|r)>0$ such that $\hat{h}_{A}(s^{*}|r)=1$ for $\hat{h}_{A}\in S^{*}[b,r;g_{D}]$.
Suppose in contrast that there exists a $s^{*}\in W^{\natural}[r]$ such that $\hat{h}_{A}(s^{*}|r)=0$, for a $\hat{h}_{A}\in S^{*}[b,r;g_{D}]$.
Then, there exists $\hat{s}$ with $\hat{h}_{A}(\hat{s}|r)=1$ such that, for all $h'_{A}:\Gamma \mapsto\Delta(W)$, 
\[
\begin{aligned}
    &\sum\nolimits_{s}u_{A}(s,b)\hat{h}_{A}(s|r)\mu_{\sigma}(b|r) \\
    &= u_{A}(\hat{s},b)\mu_{\sigma}(b|r)\leq u_{A}(s^{*}, b)h'_{A}(s^{*}|r)\mu_{\sigma}(b|r)+ \sum\nolimits_{s: s\neq s^{*}} u_{A}(s, b)h'_{A}(s|r)\mu_{\sigma}(b|r),
\end{aligned}
\]
where the equality of the inequality holds when $h'_{A}=\hat{h}_{A}$.
For all $h'_{A}\neq\hat{h}_{A}$, $h'_{A}(s^{*}|r)\in[0,1]$, which implies $u_{A}(\hat{s},b)\mu_{\sigma}(b|r)< u_{A}(\hat{s},b)\mu_{\sigma}(b|r)$ for all $b\in W$ and $r\in\Gamma$ with $\mu_{\sigma}(b|r)>0$.
Thus, $\sum_{b}u_{A}(\hat{s},b)\mu_{\sigma}(b|r)< \sum_{b}u_{A}(s^{*},b)\mu_{\sigma}(b|r)$, which contradicts to $s^{*}\in W^{\natural}[r]$.
Therefore, $W^{\natural}[r]=\cup_{b}\{s\in W| h_{A}(s|r)=1, h_{A}\in S^{*}[b,r;g_{D}]\}$.
It is not hard to see that every feasible mixed strategy $h_{A}(\cdot|r)$ that assigns strictly positive probability only to elements of $W^{\natural}[r]$ is a best response to $g_{D}$.
Since $h^{\mu}_{A}\in S^{*}[b,r;g_{D}]$, we can conclude that $h^{*}_{A}\in \mathcal{BR}^{\sigma}[g_{D}]$ with $h^{*}_{A}(s|r)=h^{\mu}_{A}(s|r)$, for all $s\in W$, $r\in\Gamma$.
In addition, we can rewrite $Z(g_{D},\sigma)$ in terms of $h^{\mu}_{A}$ as $Z(g_{D}, \sigma)=\sum\nolimits_{b,s}\int_{r}v(s,b)h^{\mu}_{A}(s|r)\rho_{D}(r|b)dr q(b)$.
Thus, by Proposition \ref{prop:mu_least_privacy}, we conclude that $L^{\sigma}(g_{D})=Z(g_{D},\sigma)$.

\subsection{Proof of Theorem \ref{thm:Bayesian_better_LRT} (Cont'd)}



\begin{lemma}\label{lemma:MuAlsoBestGamma}
    Fix any $g_{D}$. 
    Suppose that $\sigma$ is aligned.
    Let $h^{\mu}_{A}: \Gamma\mapsto\Delta(W)$ be defined by $h^{\mu}_{A}(s|r)=\mu_{\sigma}(b|r)\mathbf{1}(s=b)$ for all $s,b\in W$, $r\in\Gamma$.
    Then, $h^{\mu}_{A}\in \mathcal{BR}^{\sigma}_{\Gamma}[g_{D}]$.
\end{lemma}

\begin{proof}
Let $V^{\ddagger}_{A}( s, r)\equiv \sum\nolimits_{b}  u_{A}(s,b)\mu_{\sigma}(b|r)$
Define $W^{\ddagger}[r]\equiv\{s\in W| s\in\argmin_{s'}V^{\ddagger}_{A}( s, r)\}$.
Hence, each $h_{A}:\Gamma\mapsto\Delta(W)$ that only assigns strictly positive probabilities to $s\in W^{\ddagger}[r]$ satisfies $h_{A}\in \mathcal{BR}^{\sigma}_{\Gamma}[g_{D}]$.
In addition, let $W^{\sharp}[r]\equiv\{s\in W|\mu_{\sigma}(s|r)>0\}$.
By definition of $c_{A}$ and $v_{s,b}$, $\gamma c_{A}(s) - v_{A}(s,b)$ (weakly) decreases as $\sum_{k\in U}\mathbf{1}\{s_{k}=b_{k}\}$ increases.
Thus, $V^{\ddagger}_{A}(s^{\sharp},r)\leq V^{\ddagger}_{A}(s,r)$ for all $s^{\sharp}\in W^{\sharp}[r]$ and $s\in W$.
Hence, $W^{\sharp}[r]\subseteq W^{\ddagger}[r]$ for all $r$.
Hence, $h^{\mu}_{A}\in \mathcal{BR}^{\sigma}_{\Gamma}[g_{D}]$ holds. 
\end{proof}

With abuse of notation, we let $q(b)$ and $q(b_{k})=\sum_{b_{-k}}q(b_{k}, b_{-k})$ denote the prior and the marginalized prior, respectively.
Next, we show that optimal $\alpha$-LRT cannot strictly outperform $\sigma$-Bayesian under the same $g_{D}$.

\begin{lemma}\label{lemma:LRTlessTruePrior}
    Fix $g_{D}$ and $\alpha$. Suppose $\sigma=q$. Then, $Z(g_{D},q)\geq L(g_{D},\tau^{*},\alpha)$.
\end{lemma}

\begin{proof}
    Suppose in contrast that $Z(g_{D},q)< L(g_{D},\tau^{*},\alpha)$. Then,
\begin{equation*}
    \begin{aligned}
        \sum\nolimits_{k}P^{k}_{1}\left[y_{k}(r,\tau^{*})=1\middle|g_{D}\right] q(b_{k}=1) &>\sum\nolimits_{b,s}\int_{r}v(s,b)\mu_{\sigma}(s|r)\rho_{D}(r|b)dr q(b)\\
&=\sum\nolimits_{k}P^{k}_{\sigma}\left[s_k=1\middle|g_{D},b_{k}=1\right]q(b_{k}=1),
    \end{aligned}
\end{equation*}
where $P^{k}_{\sigma}\left[s_k=1\middle|g_{D},b_{k}=1\right]=\int_{r} \mu_{\sigma}(s_{k}=1|r)\rho_{D}(r|b)dr$.
From Proposition \ref{thm:attacker_mirror_mu}, we have $ L^{\sigma}(g_{D})= Z(g_{D},q)< L(g_{D},\tau^{*},\alpha)$.
Let $h^{\dagger}_{A}(s_{k}=1|r)=\mathbf{1}\left\{y_{k}(r,\tau^{*})=1\right\}$ for all $r\in\Gamma$.
Since $\sigma=q$, $h^{\dagger}_{A}$ is a best response of the Bayesian attacker.
Hence, $L(g_{D},\tau^{*},\alpha)=L(g_{D}, h^{\dagger}_{A})\leq Z(g_{D},q)$, which contradicts to $Z(g_{D},q)< L(g_{D},\tau^{*},\alpha)$.
Therefore, $Z(g_{D},q)\geq L(g_{D},\tau^{*},\alpha)$.
\end{proof}

If $\sigma$ is informative, we have $L(g_{D}, h^{\sigma}_{A},\sigma)\leq L(g_{D}, h^{q}_{A},q)$.
Hence, it also holds that $Z(g_{D},\sigma)\geq Z(g_{D}, q)$.
Lemma \ref{lemma:LRTlessTruePrior} imples $Z(g_{D},\sigma)\geq L(g_{D},\tau^{*},\alpha)$.

Next, we show that when $\sigma$ is non-informative. Let $h^{\sigma}_{A}(s|r)=\mu_{\sigma}(b|r)\mathbf{1}\{s=b\}$, for all $s,b\in W$, $r\in\Gamma$.
By Lemma \ref{lemma:MuAlsoBestGamma}, it holds that $h^{\sigma}_{A}\in\mathcal{BR}^{\sigma}_{\Gamma}[g_{D}]$.
Suppose in contrast that $L(g_{D}, \tau^{*}, \alpha)> Z(g_{D},\sigma)$.
Then, $h^{\sigma}_{A}\in\mathcal{BR}^{\sigma}_{\Gamma}[g_{D}]$ implies
\[
        \sum\nolimits_{k}P^{k}_{1}\left[y_{k}(r,\tau^{*})=1\middle|g_{D}\right] >\sum\nolimits_{k,s}\int_{r}v(s_{k}=1,b_{k}=1)\mu_{\sigma}(s_{k}=1|r).
        \]
Let $h^{\dagger}_{A}:\Gamma\mapsto\Delta(W)$ such that $h^{\dagger}_{A}(s_{k}=1|r)=\mathbf{1}\left\{y_{k}(r,\tau^{*})=1\right\}$ for all $r\in\Gamma$.
Then, $h^{\dagger}_{A}\in\mathcal{BR}^{\sigma}[g_{D}]$ when $\sigma$ is uniform (i.e., non-informative).
Proposition \ref{thm:attacker_mirror_mu} implies $Z(g_{D},\sigma)\geq L(g_{D}, h^{\dagger}_{A}, \sigma)$, which leads to a contradiction.
The inequality $L(g_{D}, \tau^{*}, \alpha)\geq L(g_{D}, \tau^{(N)}(\tilde{r}),\alpha_{\tau^{(N)}(\tilde{r})})$ follows the Neyman-Pearson lemma.
In addition, by \cite{venkatesaramani2021defending,venkatesaramani2023enabling}, $L(g_{D}, \tau^{(N)}(\tilde{r}),\alpha_{\tau^{(N)}(\tilde{r})})\geq L(g_{D}, \tau^{o},\alpha_{\tau^{o}})$.
Thus, we can conclude the proof of Theorem \ref{thm:Bayesian_better_LRT}.
\qed

\section{Proof of Lemma \ref{lemma:normal_relationship}}

First, we show that the test statistics $\mathcal{L}(\tilde{y})=\sum\nolimits_{j\in Q}\log\left(\rho_{j}(\tilde{y}_{j}|\widehat{H}_{0})/\rho_{j}(\tilde{y}_{j}|\widehat{H}_{1})\right)$ is normally distributed under $\widehat{H}_{0}$ and $\widehat{H}_1$, respectively, with $\mathcal{N}\left(\overline{\textup{M}}, \overline{\textup{V}}\right)$ and $\mathcal{N}\left(-\overline{\textup{M}}, \overline{\textup{V}}\right)$, where $\overline{\textup{M}}=\frac{1}{2}\sum_{j\in Q}\widehat{\mathtt{M}}^{2}_{j}$ and $\overline{\textup{V}}=\sum_{j\in Q}\widehat{\mathtt{M}}^{2}_{j}$.
For each $y_{j}$, $\tilde{y}_{j}\sim\mathcal{N}(0,1)$ under $\widehat{H}_{0}$, and $\tilde{y}_{j}\sim\mathcal{N}(\widehat{M}_{j},1)$ under $\widehat{H}_{1}$. Thus, the log-likelihood ratio for each $y_{j}$ is $\log\left(\frac{\rho_{j}(y_{j}|\widehat{H}_{0})}{\rho_{j}(y_{j}|\widehat{H}_{1})}\right)$.
Since $\rho_{j}(\cdot|\widehat{H}_{0})$ and $\rho_{j}(\cdot|\widehat{H}_{1})$ are the density functions of normal distribution, the log-likelihood ratio becomes
\[
\begin{aligned}
    \log \left(\frac{\frac{1}{\sqrt{2\pi}} e^{-\frac{y_j^2}{2}}}{\frac{1}{\sqrt{2\pi}} e^{-\frac{(y_j - \widehat{\mathtt{M}}_j)^2}{2}}}\right) = \frac{(y_j - \widehat{\mathtt{M}}_j)^2 - y_j^2}{2}=\frac{-2y_j\widehat{\mathtt{M}}_j + \widehat{\mathtt{M}}_j^2}{2} = -y_j\widehat{\mathtt{M}}_j + \frac{\widehat{\mathtt{M}}_j^2}{2}.
\end{aligned}
\]

Under $\widehat{H}_{0}$, the mean is $\mathbb{E}[y_{j}|\widehat{H}_{0}]=0$ and the variance is $\textup{Var}[y_j|\widehat{H}_{0}]=1$. Hence, the mean of $\mathcal{L}(y)$ under $\widehat{H}_{0}$ is
\[
\begin{aligned}
    \mathbb{E}[\mathcal{L}(y)] = \mathbb{E}\left[\sum_{j\in Q} \left(-y_j \widehat{\mathtt{M}}_j + \frac{\widehat{\mathtt{M}}_j^2}{2}\right)\right] = \sum_{j\in Q} \left(-\mathbb{E}[y_j] \widehat{\mathtt{M}}_j + \frac{\widehat{\mathtt{M}}_j^2}{2}\right) = \sum_{j\in Q} \frac{\widehat{\mathtt{M}}_j^2}{2},
\end{aligned}
\]
and the variance is 
\[
\begin{aligned}
    \text{Var}[\mathcal{L}(y)] = \text{Var}\left[\sum_{j\in Q} \left(-y_j \widehat{\mathtt{M}}_j + \frac{\widehat{\mathtt{M}}_j^2}{2}\right)\right] = \sum_{j\in Q} \text{Var}[-y_j \widehat{\mathtt{M}}_j] = \sum_{j\in Q} \widehat{\mathtt{M}}_j^2.
\end{aligned}
\]

Similarly, under $\widehat{H}_{1}$, the mean of $\mathcal{L}(y)$ is
\[
\begin{aligned}
    \mathbb{E}[\mathcal{L}(y)] = \mathbb{E}\left[\sum_{j\in Q} \left(-y_j \widehat{\mathtt{M}}_j + \frac{\widehat{\mathtt{M}}_j^2}{2}\right)\right]=\sum_{j\in Q} \left(-\mathbb{E}[y_j] \widehat{\mathtt{M}}_j + \frac{\widehat{\mathtt{M}}_j^2}{2}\right) = \sum_{j\in Q} \left(-\widehat{\mathtt{M}}_j^2 + \frac{\widehat{\mathtt{M}}_j^2}{2}\right)=\sum_{j\in Q} -\frac{\widehat{\mathtt{M}}_j^2}{2}.
\end{aligned}
\]
In addition, the variance of $\mathcal{L}(y)$ under $\widehat{H}_{1}$ is
\[
\begin{aligned}
    \text{Var}[\mathcal{L}(y)] = \text{Var}\left[\sum_{j\in Q} \left(-y_j \widehat{\mathtt{M}}_j + \frac{\widehat{\mathtt{M}}_j^2}{2}\right)\right] = \sum_{j\in Q} \text{Var}[-y_j \widehat{\mathtt{M}}_j] = \sum_{j\in Q} \widehat{\mathtt{M}}_j^2.
\end{aligned}
\]

%

Since the test statistics $\mathcal{L}(y)$ is normally distributed under $\widehat{H}_{0}$ and $\widehat{H}_{1}$, we have
\[
Z_{0}=\frac{\overline{y}-\overline{\textup{M}}}{\sqrt{\overline{\textup{V}}/\sqrt{m}}} \sim \mathcal{N}(0,1) \textup{ and } Z_{1}= \frac{\overline{y}+\overline{\textup{M}}}{\sqrt{\overline{\textup{V}}/\sqrt{m}}} \sim \mathcal{N}\left(\frac{-2\overline{M} }{\sqrt{\overline{\textup{V}}/\sqrt{m}}} ,1\right),
\]
where $\overline{y}$ is the sample mean.
For a given significance level $\widehat{\alpha}$, the threshold for $Z_{0}$ is set so that $\textup{Pr}(Z_{0}< z_{\widehat{\alpha}}) = \widehat{\alpha}$, corresponding to the value $\overline{\textup{M}} + z_{\widehat{\alpha}}\sqrt{\frac{\overline{\textup{V}}}{m}}$.
For a given Type-II error rate $\widehat{\beta}$, the threshold for $Z_{1}$ is set so that $\textup{Pr}(Z_{1}<z_{\widehat{\beta}})=\widehat{\beta}$, where $z_{\widehat{\beta}}$ aligns with $-\overline{\textup{M}} - z_{\widehat{\beta}}\sqrt{\frac{\overline{\textup{V}}}{m}}$.
To maintain the consistency of decision-making between $\widehat{H}_{0}$ and $\widehat{H}_{1}$, the threshold at which we switch decisions from failing to reject $\widehat{H}_{0}$ to rejecting $\widehat{H}_{0}$ under $\widehat{H}_{0}$ and $\widehat{H}_{1}$ are equated. 
Therefore, we have 
\begin{equation*}
    \sqrt{m}\overline{\textup{M}}+ z_{\widehat{\alpha}}\sqrt{\overline{\textup{V}}}= -\sqrt{m}\overline{\textup{M}}- z_{\widehat{\beta}}\sqrt{\overline{\textup{V}}}.
\end{equation*}
Thus, $\mathcal{F}\left(\widehat{\alpha}, \widehat{\beta}\right)=m$ holds.

Next, we show the monotone relationship between $\widehat{\beta}$ and $m$ given $\mathcal{F}\left(\widehat{\alpha}, \widehat{\beta}\right)=m$  while everything else is fixed.
Since, $z_{\widehat{\beta}}=\Phi^{-1}(1-\widehat{\beta})$, where $\Phi$ is the cumulative distribution function (CDF) of the standard normal distribution, $z_{\widehat{\beta}}$ decreases as $\widehat{\beta}$ increases as the quantile function $\Phi^{-1}$ decreases as the probability increases. 
As a result, $(z_{\widehat{\alpha}}, z_{\widehat{\beta}})$ decreases when $\widehat{\beta}$ increases.
Therefore, $\mathcal{F}\left(\widehat{\alpha}, \widehat{\beta}\right)=m$ implies that $m$ decreases when $\widehat{\beta}$ increases.
\qed

\section{Proof of Theorem \ref{thm:Gaussian_comparison}}\label{app:proof_thm2}

We first obtain the following lemma, which extends Theorem 2.7 of \cite{dong2021gaussian}.

\begin{lemma}\label{lemma:gpd}
Fix $\alpha\in(0,1)$.
    Let $g_{D}$ be Gaussian defined above with each $g^{j}_{D}(\cdot|b)\in\Delta(\mathcal{Y}_{j})$ as the density function of $\mathcal{N}(\mathtt{M}^{j}_{b}, \mathtt{V}^{j})$ given any $b\in W$, where $\mathtt{V}^{j}=\left(2\textup{sens}^{j}(f)/\widehat{\mathtt{M}}_{j}\right)^{2}$.
    Let $P_{b}(g^{j}_{D})$ denote the probability distribution associated with $g^{j}_{D}(\cdot|b)$.
    Suppose $\max_{b,b'}|\mathtt{M}^{j}_{b}-\mathtt{M}^{j}_{b'}|\leq \textup{sens}^{j}(f)$.
    Then, it holds
    \[
    T\left[P_{b}(g^{j}_{D}), P_{b'}(g^{j}_{D})\right](\alpha)\geq T\left[\mathcal{N}(0,1), \mathcal{N}(\widehat{\mathtt{M}}_{j}, 1)\right].
    \]
\end{lemma}

\begin{proof}

For any two $b,b'\in W$, $y(b)=f_{j}(b,d) + \delta_{j}$ and $y(b')=f_{j}(b',d') + \delta'_{j}$ are normally distributed with means $f^{j}(b,d)+\mathtt{M}_{b}$ and $f^{j}(b',d')+\mathtt{M}_{b'}$, respectively, and a common variance $\mathtt{V}^{j}$.
Then, we have
\[
\begin{aligned}
    T\left[P_{b}(g^{j}_{D}), P_{b'}(g^{j}_{D})\right](\alpha)=&T\left[\mathcal{N}\left(f^{j}(b,d)+\mathtt{M}_{b}, \mathtt{V}\right), \mathcal{N}\left(f^{j}(b',d')+\mathtt{M}_{b'}, \mathtt{V}\right) \right](\alpha)\\
    =&\Phi\left(\Phi^{-1}\left(1-\alpha\right) -\frac{|f^{j}(b,d)- f^{j}(b',d') + \mathtt{M}_{b}-\mathtt{M}_{b'}|}{\sqrt{\mathtt{V}^{j}}} \right),
\end{aligned}
\]
where $\Phi$ is the cumulative distribution function (CDF) of the standard normal distribution.
Since $\mathtt{V}^{j}=\left(2\textup{sens}^{j}(f)/\widehat{\mathtt{M}}_{j}\right)^{2}$ and $\max_{b,b'}|\mathtt{M}^{j}_{b}-\mathtt{M}^{j}_{b'}|\leq \textup{sens}^{j}(f)$, by definition of sensitivity, we obtain
\[
\begin{aligned}
    T\left[\mathcal{N}\left(f^{j}(b,d)+\mathtt{M}_{b}, \mathtt{V}\right), \mathcal{N}\left(f^{j}(b',d')+\mathtt{M}_{b'}, \mathtt{V}\right) \right](\alpha)&\geq \Phi\left(\Phi^{-1}\left(1-\alpha \right)- \widehat{\mathtt{M}}_{j}\right)\\
    &=T\left[\mathcal{N}(0,1), \mathcal{N}(\widehat{\mathtt{M}}_{j}, 1)\right](\alpha).
\end{aligned}
\]
\end{proof}

Lemma \ref{lemma:gpd} shows that distinguishing between $b$ and $b'$ is as hard as distinguishing between $\mathcal{N}(0,1)$ and $ \mathcal{N}(\widehat{\mathtt{M}}_{j}, 1)$.
Thus, if the $\alpha$-LRT attacker only observes $y_{j}$ for $j$th SNV, then the maximum power he can obtain is $1-T\left[\mathcal{N}(0,1), \mathcal{N}(\widehat{\mathtt{M}}_{j}, 1)\right](\alpha)$, which leads to the WCPL for the vNM defender among all possible powers when different membership vectors are realized.
considered are independent, $1-T\left[\mathcal{N}(0,1), \mathcal{N}(\widehat{\mathtt{M}}_{j}, 1)\right](\alpha)$ serves as the performance bound for every $j\in Q$.

Given any two $b,b'\in W$, define the hypothesis testing problem: $H_{0}: \textup{ the membership vector is }b$ versus $H_{1}: \textup{ the membership vector is }b'$.
From the assumption of independent SNVs, we can obtain the log-likelihood statistics 
\[
\ell(y;g_{D}, b, b')\equiv\sum\nolimits_{j\in Q} \log\left(\frac{\rho^{j}_{D}(y_{j}|H_0)}{\rho^{j}_{D}(y_{j}|H_1)} \right).
\]
Let $P_{i}[\cdot|g_{D}]$ denote the probability distribution associated with $H_{i}$ for $i\in\{0,1\}$.

\begin{lemma}
Fix $\alpha\in(0,1)$.
    Let $g_{D}$ be Gaussian defined above with each $g^{j}_{D}(\cdot|b)\in\Delta(\mathcal{Y}_{j})$ as the density function of $\mathcal{N}(\mathtt{M}^{j}_{b}, \mathtt{V}^{j})$ given any $b\in W$, where $\mathtt{V}^{j}=\left(2\textup{sens}^{j}(f)/\widehat{\mathtt{M}}_{j}\right)^{2}$.
    Suppose $\max_{b,b'}|\mathtt{M}^{j}_{b}-\mathtt{M}^{j}_{b'}|\leq \textup{sens}^{j}(f)$.
    Then, it holds for all pair $b,b'\in W$,
    \begin{equation}\label{eq:composition}
        \max_{\tau}P_{1}\left[\ell(\tilde{y};g_{D}, b, b')\geq \tau\middle|g_{D}\right] \leq 1-T\left[\mathcal{N}(0,1), \mathcal{N}\left(\sqrt{\sum\nolimits_{j\in Q}\widehat{\mathtt{M}}^{2}_{j}}, 1\right)\right](\alpha),
    \end{equation}
    with $P_{0}\left[\ell(\tilde{y};g_{D}, b, b')< \tau\middle|g_{D}\right]=\alpha$.
\end{lemma}

\begin{proof}
Since the SNVs are independent, the joint probability density $P(y|H_{i})$ over $\mathcal{Y}$ that is equal to the product $\prod_{j\in Q} \rho^{j}_{D}(y_{j}|H_{i})$ for $i\in\{0,1\}$.
It is a $|Q|$-fold composition of $\{\rho^{j}_{D}\}_{j\in Q}$, where each $\rho^{j}_{D}$ accesses to the same dataset.
In addition, $\max_{\tau}P_{1}\left[\ell(\tilde{y};g_{D}, b, b')\geq \tau\middle|g_{D}\right]$ is the power of $\alpha$-LRT given $g_{D}$ for any $b,b'\in W$.
Then, (\ref{eq:composition}) follows Corollary 3.3 of \cite{dong2021gaussian}.
\end{proof}

Let $\mathbf{I}_{|Q|}$ denote a $|Q|\times |Q|$ identity matrix. Let $\widehat{\mathbf{M}}\equiv(\widehat{\mathtt{M}}_{1}, \dots, \widehat{\mathtt{M}}_{|Q|})$.
Consider two multivariate normal distribution $\mathcal{N}(0, \mathbf{I}_{|Q|})$ and $\mathcal{N}(\widehat{\mathbf{M}}, \mathbf{I}_{|Q|})$.
Here, $\mathcal{N}(0, \mathbf{I}_{|Q|})$ is rotation invariant, and $\mathcal{N}(\widehat{\mathbf{M}}, \mathbf{I}_{|Q|})$ can be rotated to $\mathcal{N}\left(\sqrt{\sum\nolimits_{j\in Q}\widehat{\mathtt{M}}^{2}_{j}}, 1\right)$. 
In addition, the rotation here is an invertible transformation.
Therefore, $T\left[\mathcal{N}(0,1), \mathcal{N}\left(\sqrt{\sum\nolimits_{j\in Q}\widehat{\mathtt{M}}^{2}_{j}}, 1\right)\right](\alpha)$ is the same as the $T\left[\mathcal{N}(0, \mathbf{I}_{|Q|}), \mathcal{N}(\widehat{\mathbf{M}}, \mathbf{I}_{|Q|})\right](\alpha)$ for any $\alpha$ because the trade-off function is invariant under invertible transformations \cite{dong2021gaussian}.
Let $\widehat{\beta}=T\left[\mathcal{N}(0, \mathbf{I}_{|Q|}), \mathcal{N}(\widehat{\mathbf{M}}, \mathbf{I}_{|Q|})\right](\alpha)$.
Thus, the $\alpha$-LRT with the LR statistics formulated by $\mathcal{L}(y)$ has the power $1-\widehat{\beta}$.
Therefore, it holds that $\mathcal{F}(\alpha, \widehat{\beta}) = |Q|$.

Now, let us focus on when the attacker (either Bayesian or LRT) targets a specific individual $k$.
Given any subjective prior $\sigma$ and $Q$, let $\mu^{\sigma}_{1|0}[|Q|]=\int_{r}\mu_{\sigma}(s_{k}=1|r)\rho_{D}(r|b_{k}=0)$.
By Proposition \ref{thm:attacker_mirror_mu}, a Bayesian attacker's strategy that mirrors the distribution of the posterior belief leads to the WCPL for the defender.
Hence, $\mu^{\sigma}_{1|0}[|Q|]$ captures the highest Type-II errors of the Bayesian attacker. 
Then, $\mathcal{F}\left(\alpha, \mu^{\sigma}_{0|1}[|Q|]\right)$ captures the number of SNVs (i.e., $|Q|$) so that $\alpha$-LRT can attain the power $\mu^{\sigma}_{0|1}[|Q|]$ when the set $Q$ of SNVs of each individual are used in the dataset, leading to $L(g_{D}, \tau^{*},\alpha) = L^{\sigma}(g_{D})$.
If $\mathcal{F}\left(\alpha, \mu^{\sigma}_{0|1}[|Q|]\right)\geq|Q|$, then more SNVs needs to be used to make $\alpha$-LRT have the power $\mu^{\sigma}_{0|1}[|Q|]$.
This is equivalent to $\widehat{\beta}<\mu^{\sigma}_{0|1}[|Q|]$, which implies $L(g_{D}, \tau^{*},\alpha)\leq L^{\sigma}(g_{D})$.
\qed

\section{Experiment Details}\label{sec:app_Hyperparameters}

\FloatBarrier

\subsection{Dataset}

The dataset used in our experiments was initially provided by the organizers of the 2016 iDash Privacy and Security Workshop \cite{tang2016idash} as part of their challenge on Practical Protection of Genomic Data Sharing Through Beacon Services. 
In this research, we follow \cite{venkatesaramani2021defending,venkatesaramani2023enabling} and employ SNVs from chromosome $10$ for a subset of $400$ individuals to construct the Beacon, with another $400$ individuals excluded from the Beacon.

\subsection{Network Configurations and Hyperparameters}

The \textbf{Defender} neural network is a generative model designed to process membership vectors and produce beacon modification decisions. The input layer feeds into two fully connected layers with batch normalization and activation functions applied after each layer. The first hidden layer uses ReLU activation, while the second hidden layer uses LeakyReLU activation. The output layer applies a scaled sigmoid activation function. The output of the Defender neural network is a real value between -0.5 and 0.5, which is guaranteed by the scaled sigmoid activation function.
All Defender neural networks were trained using the Adam optimizer with a learning rate of 0.001, weight decay of $0.00001$, and an ExponentialLR scheduler with a decay rate of $0.988$.

The \textbf{Attacker} neural network is a generative model designed to process beacons and noise to produce membership vectors. The input layer feeds into two fully connected layers with batch normalization and activation functions. The first hidden layer uses ReLU activation. The output layer applies a sigmoid activation function.
All Attacker models were trained using the Adam optimizer, a learning rate of $0.0001$, weight decay of $0.00001$, and an ExponentialLR scheduler with a decay rate of $0.988$.

The specific configurations for each model are provided in the tables below. Table \ref{Table:defender_scalar} shows the configurations of the neural network Defender under the Bayesian, the fixed-threshold, and the adaptive-threshold attackers when the trade-off parameter $\kappa$ is a vector (i.e., each $\kappa_j=\kappa$ for all $j\in Q$).
Table \ref{Table:defender_vector} shows the configurations of Defender when the trade-off parameter is a vector; i.e.,  $\vec{\kappa}=(\kappa_{j})_{j\in Q}$ where $\kappa_{j}=0$ for the $90\%$ of $5000$ SNVs and $\kappa_{j}=50$ for the remaining $10\%$.
Table \ref{Table:attacker_defender} lists the configurations of the neural network Attacker under the Bayesian, the fixed-threshold LRT, and the adaptive-threshold LRT defenders.
Table \ref{Table:attacker_DP}
lists the configurations of Attacker under the standard $\epsilon$-DP which induces the same $\vec{\kappa}$-weighted expected utility loss for the defender.

\begin{table}[H]
\centering
\caption{Bayesian Defender Configurations}\label{Table:defender}
\begin{minipage}{0.45\textwidth}
\centering
\subcaption{Defender with scalar $\kappa$}\label{Table:defender_scalar}
\begin{tabular}{|l|c|c|}
\hline
Layer            & Input Units & Output Units \\ \hline
Input Layer      & 830         & 1500         \\ \hline
Hidden Layer 1   & 1500        & 1100         \\ \hline
Hidden Layer 2   & 1100        & 500          \\ \hline
Output Layer     & 500         & 5000         \\ \hline
\end{tabular}
\end{minipage}
\hspace{0.05\textwidth}
\begin{minipage}{0.45\textwidth}
\centering
\subcaption{Defender with vector $\vec{\kappa}$}\label{Table:defender_vector}
\begin{tabular}{|l|c|c|}
\hline
Layer            & Input Units & Output Units \\ \hline
Input Layer      & 830         & 1000         \\ \hline
Hidden Layer 1   & 1000        & 3000         \\ \hline
Hidden Layer 2   & 3000        & 4600         \\ \hline
Output Layer     & 4600        & 5000         \\ \hline
\end{tabular}
\end{minipage}
\end{table}

\begin{table}[H]
\centering
\caption{Attacker Configurations}
\begin{subtable}[t]{0.45\textwidth}
\centering
\caption{Attacker vs. Defender}\label{Table:attacker_defender}
\begin{tabular}{|l|c|c|}
\hline
Layer            & Input Units & Output Units \\ \hline
Input Layer      & 5000        & 3400         \\ \hline
Hidden Layer 1   & 3400        & 2000         \\ \hline
Output Layer     & 2000        & 800          \\ \hline
\end{tabular}
\end{subtable}
\hspace{0.05\textwidth}
\begin{subtable}[t]{0.45\textwidth}
\centering
\caption{Bayesian Attacker vs. $\epsilon$-DP}\label{Table:attacker_DP}
\begin{tabular}{|l|c|c|}
\hline
Layer            & Input Units & Output Units \\ \hline
Input Layer      & 5000        & 3000         \\ \hline
Hidden Layer 1   & 3000        & 1000         \\ \hline
Output Layer     & 1000        & 800          \\ \hline
\end{tabular}
\end{subtable}
\end{table}

\FloatBarrier

\subsection{AUC Values of ROC Curves with Standard Deviations}

Tables \ref{Table:AUC_different_attacks} and \ref{Table:AUC_different_defender} show the AUC values of the ROC curves shown in the plots of Figure \ref{fig:six_figures} in the experiments.

\begin{table}[h!]
\centering
\caption{AUC Values For Different Attackers Under Varying $\kappa$}\label{Table:AUC_different_attacks}
\begin{tabular}{@{}lccc@{}}
\toprule
Attacker           & Figure 1a ($\kappa=0$) & Figure 1b ($\kappa=1.5$) & Figure 1c ($\kappa=50$) \\ \midrule
Bayesian attacker  & $0.5205 \pm 0.0055$    & $0.7253 \pm 0.0069$      & $0.8076 \pm 0.0040$     \\
Fixed-Threshold LRT attacker   & $0.5026 \pm 0.0062$    & $0.6214 \pm 0.0322$      & $0.7284 \pm 0.0089$     \\
Adaptive-Threshold LRT attacker   & $0.1552 \pm 0.0100$    & $0.1716 \pm 0.0144$      & $0.1719 \pm 0.0174$     \\ \bottomrule
\end{tabular}
\end{table}

\begin{table}[h!]
\centering
\caption{AUC Values of Attackers For Figures 1d to 1f }\label{Table:AUC_different_defender}
\begin{tabular}{@{}cccc@{}}
\toprule
Figure & Scenarios             & AUC $\pm$ std        & Condition \\ \midrule
\multirow{3}{*}{1d} & Under Bayesian Defender            & $0.7237 \pm 0.0066$ & $\kappa = 1.5$ \\
                    & Under Fixed-threshold LRT Defender & $0.9124 \pm 0.0026$ & $\kappa = 1.5$ \\
                    & Under Adaptive-threshold LRT Defender   & $0.7487 \pm 0.0027$ & $\kappa = 1.5$ \\ \midrule
\multirow{3}{*}{1e} & Bayesian Attacker            & $0.7178 \pm 0.0050$ & $\epsilon = 600$ \\
                    & Fix-LRT Attacker             & $0.6285 \pm 0.0057$ & $\epsilon = 600$ \\
                    & Adp-LRT Attacker             & $0.2402 \pm 0.0117$ & $\epsilon = 600$ \\ \midrule
\multirow{2}{*}{1f} & Under Bayesian Defender            & $0.5318 \pm 0.0222$ & $\vec{\kappa}$ \\
                    & Under $\epsilon$-DP Defender       & $0.9153 \pm 0.0025$ & $\vec{\kappa}$ \\ \bottomrule
\end{tabular}
\end{table}

\end{document}